\documentclass[12pt]{article}
\usepackage{amsmath,amsfonts,amssymb,amsthm}

 \usepackage
 {hyperref}

\usepackage{enumerate}

\hypersetup{
   colorlinks = true,
  linkcolor = red,
   anchorcolor = red,
   citecolor = blue,
   filecolor = red,
   urlcolor = red
}

\def\v #1{\vert #1\vert}             

\def\n #1 #2{(-1)^{{\v #1} {\v #2}}} 

\theoremstyle{plain}
\newtheorem{theorem}{Theorem}
\newtheorem{corollary}[theorem]{Corollary}
\newtheorem{proposition}[theorem]{Proposition}
\newtheorem{lemma}[theorem]{Lemma}

\theoremstyle{definition}
\newtheorem{definition}[theorem]{Definition}

\numberwithin{equation}{section}
\numberwithin{theorem}{section}

\def\<#1>{\langle#1\rangle}

\begin{document}
\centerline{\Large \bf Lie symmetries for Lie systems:}\vskip 0.35cm
\centerline{\Large \bf applications to systems of ODEs and PDEs} \vskip 0.7cm

\centerline{P.G. Est\'evez$^{a)}$, F.J. Herranz$^{b)}$, J. de Lucas$^{c)}$ and C.
Sard\'on$^{a)}$}
\vskip 0.5cm
\centerline{$^{a)}$Department of Fundamental Physics, University of
Salamanca,}
\centerline{Plza. de la Merced s/n, 37.008, Salamanca, Spain.}
\vskip 0.5cm
\centerline{$^{b)}$Department of Mathematical Methods in Physics, University of Warsaw,}
\centerline{ul. Hoza 74, 00-682, Warszawa, Poland.}
\vskip 0.5cm
\centerline{$^{c)}$Department of Physics, University of Burgos,}
\centerline{09001, Burgos, Spain.}

\vskip 1cm

\begin{abstract}
A {\it Lie system} is a nonautonomous system of first-order differential equations  admitting a {\it superposition rule}, i.e., a map expressing its general solution
in terms of a generic family of particular solutions and some constants. Using that a Lie system can be considered as a curve in a finite-dimensional Lie algebra of vector fields, a so-called {\it Vessiot--Guldberg Lie algebra}, we associate every Lie system with a Lie algebra of Lie point symmetries induced by the Vessiot--Guldberg Lie algebra. This enables us to derive Lie symmetries of relevant physical systems described by first- and higher-order systems of differential equations by means of Lie systems in an easier way than by standard methods. A generalization of our results to partial differential equations is introduced. Among other  applications, Lie symmetries for several new and known generalizations of the real Riccati equation are studied.
\end{abstract}

\section{Introduction}

The analysis of Lie systems dates back to the end of the XIXth century, when Vessiot, Guldberg and Lie \cite{Gu93,LS,Ve93} discovered the most fundamental properties of nonautonomous systems of first-order ordinary differential equations (ODEs) admitting a superposition rule.
Lie proved the nowadays called {\it Lie--Scheffers Theorem}, which states that a nonautonomous system of first-order ODEs admits a superposition rule if and only if it describes the integral curves of a $t$-dependent vector field taking values in a
finite-dimensional real Lie algebra $V$ of vector fields \cite{CGM07}. We call $V$ a {\it Vessiot--Guldberg Lie algebra} of the Lie system \cite{Dissertations}.

The theory of Lie systems was deeply analyzed towards the end of the XIXth century, to be almost forgotten soon after. Along the eighties of the XXth century, Winternitz and coworkers derived superposition rules for many relevant systems of differential and superdifferential equations and studied the existence of Lie systems on $\mathbb{R}^2$ among other topics \cite{BecHusWin86,PW}. 

Cari\~nena, Grabowski and Marmo characterized superposition rules as connections of zero curvature \cite{CGM07}. Other relevant results have also been derived for Lie systems admitting a Vessiot--Guldberg Lie algebra of Hamiltonian vector fields with respect to some geometric structure \cite{CGLS13,CLS13}. This fact led to the derivation of a superposition rule for Riccati equations in a simpler way by using a Casimir function of $C^\infty(\mathfrak{sl}(2,\mathbb{R})^*)$ \cite{BCHLS13}. These results illustrate that Lie systems enjoy relevant geometric properties.
Additionally, Lie systems can be applied in quantum/classical mechanics and control theory (see \cite{Dissertations} and references therein).  

The interest of Lie symmetries for the study of systems of differential equations is undoubtable. For instance, Lie symmetries allow us to write a differential equation in a simpler way. In this paper, we pioneer the study of Lie symmetries for Lie systems and their generalizations to higher-order ODEs (HODEs) and partial differential equations (PDEs): the {\it HODE and PDE Lie systems} \cite{CGM07,SIGMA}. As far as we know, only a few minor results about Lie symmetries for Lie systems had appeared before in \cite{CLS13}.
  
Given a Lie system possessing a Vessiot--Guldberg Lie algebra $V$, we study its Lie symmetries of the form
$$
Y=f_0(t)\frac{\partial}{\partial t}+\sum_{\alpha=1}^{r}f_\alpha(t)X_\alpha,
$$
where $f_0,\ldots,f_r$ are certain $t$-dependent functions and $X_1,\ldots,X_r$ form a basis for $V$.
We prove that $f_0,\ldots, f_r$ depend on the algebraic structure of $V$ and can be derived by solving a Lie system, the symmetry system of the Lie system under study, which can be expressed in coordinates in a simple way. This enables us to calculate Lie symmetries for Lie systems with isomorphic Vessiot--Guldberg Lie algebras simultaneously.

As an application, we study Lie symmetries for Riccati equations. We retrieve and generalize results given in \cite{GS10} to determine their Lie symmetries. Subsequently, we propose their generalization to the realm of split-complex and Study numbers \cite{Ro88,Yaglom}: the {\it Cayley--Klein Riccati equation}. We prove that the Cayley--Klein Riccati equation with real $t$-dependent coefficients is a Lie system admitting a Vessiot--Guldberg Lie algebra isomorphic to $\mathfrak{sl}(2,\mathbb{R})$, namely an {\it $\mathfrak{sl}(2,\mathbb{R})$-Lie system} \cite{Pi12}, and analyze its Lie symmetries with our methods. Next, we study quaternionic Riccati equations with real $t$-dependent coefficients. We prove that they are $\mathfrak{sl}(2,\mathbb{R})$-Lie systems and give a Lie system describing some of its Lie symmetries. Since all of the above mentioned Lie systems are $\mathfrak{sl}(2,\mathbb{R})$-Lie systems, our techniques enable us to obtain a common equation describing Lie symmetries for all of them. 

As a particular example, we generalize results concerning Lie symmetries for generalized and classical Darboux--Brioschi--Halphen systems \cite{Nu05}. Apart from the interest in these systems due to their applications \cite{CH03}, our example illustrates how Lie systems can be employed to investigate autonomous systems for the first time. Autonomous systems of first-order differential equations are Lie systems associated with a one-dimensional Vessiot--Guldberg Lie algebra. This implies that most of the methods to analyze Lie systems cannot be applied, as they only work properly for Vessiot--Guldberg Lie algebras of dimension two or higher (see \cite{Dissertations}). We show that this problem can be avoided by studying  autonomous systems as Lie systems related to Vessiot--Guldberg Lie algebras of dimension greater than one. Although finding such Lie algebras can be difficult, we succeeded in determining some of them for certain relevant Lie systems.   		

Subsequently, we illustrate how our results can be used to study {HODE Lie systems} \cite{CGL11}. Such systems are systems of higher-order differential equations that can be written as a Lie system by considering them as a first-order system in the usual way, namely by adding extra variables $v\equiv {\rm d}x/{\rm d}t$, $a\equiv {\rm d}v/{\rm d}t$, etc. Using this fact and the techniques of this work, we obtain {\it non-local Lie symmetries} for the initial HODE Lie system under study, where by non-local we mean Lie symmetries depending on the successive derivatives of the dependent variables. We also study Lie algebras of Lie point symmetries for HODE Lie systems, i.e., Lie symmetries that only depend on their independent and dependent variables. 

Finally, we extend our results to PDE Lie systems \cite{CGM07}. Similarly to Lie systems, PDE Lie systems can be associated with a Vessiot--Guldberg Lie algebra. Also, they can be related to a Lie algebra of Lie symmetries of a certain type which is
again determined by the Vessiot--Guldberg Lie algebra of the initial system. As before, this permits us to calculate simultaneously Lie symmetries for all PDE Lie systems possessing isomorphic Vessiot--Guldberg Lie algebras by solving another PDE Lie system. To illustrate our results, we analyze PDE Lie systems of relevance in physics and mathematics, e.g., multidimensional Riccati equations and PDEs describing certain flat connection forms \cite{FGRSZ99}. It is remarkable that the literature on PDE Lie systems is scarce (specially concerning their applications) and this paper contributes to enlarge their applications and to understand their properties \cite{Dissertations}.

The structure of the paper goes as follows. In Section 2, we briefly review the concept of Lie systems and PDE Lie systems.
Section 3 is devoted to the analysis of Lie symmetries for Lie systems. In Section 4 we provide methods to build different types of Lie algebras of Lie symmetries for Lie systems. Subsequently, we apply our results to Lie systems and HODE Lie systems of relevance in Section 5. We extend our methods to PDE Lie systems in Sections 6 and 7. In particular, we
calculate Lie symmetries for PDE Lie systems with the same symmetry system and physical and mathematical interest. In Section 8, we summarize our results and we describe our plans for future research.

\section{Fundamentals on Lie systems}

For simplicity, we hereafter assume all geometric structures to be real, smooth
and globally defined. In this way, we highlight the key points of our
presentation by omitting the analysis of minor technical problems (we refer to \cite{CGM07,Dissertations} for additional details). Subsequently, $N$ is assumed to be a real $n$-dimensional manifold with a local coordinate system $\{x_1,\ldots,x_n\}$.

Let $\tau:{\rm
T}N\rightarrow N$ be the natural projection of the tangent bundle ${\rm T}N$ onto $N$ and let
$\pi_2:(t,x)\in\mathbb{R}\times N\mapsto x\in N$ be a
projection. A {\it $t$-dependent vector field} $X$ on $N$ is a map
$X:(t,x)\in \mathbb{R}\times N\mapsto X(t,x)\in \mathbb{\rm
T}N$ such that $\tau\circ X=\pi_2$. This condition implies that every
$t$-dependent vector field $X$ on $N$ can be considered as a family
$\{X_t\}_{t\in\mathbb{R}}$ of vector fields $X_t:x\in N\mapsto
X_t(x)\equiv X(t,x)\in{\rm T}_xN\subset {\rm T}N$ on $N$ and vice versa. We call $\Gamma({\rm T}N)$ the $C^\infty(N)$-module of sections of the tangent bundle $({\rm T}N,N,\tau)$.

We call {\it integral curve} of $X$ an integral curve $\gamma:s\in\mathbb{R}\mapsto
(t(s),x(s))\in \mathbb{R}\times N$ of the {\it suspension} of $X$, i.e., the
vector field
\begin{equation}\label{xx}
\begin{array}{rccc}
\bar{X}:&\mathbb{R}\times N &\longrightarrow &{\rm T}(\mathbb{R}\times N)\simeq {\rm T}\mathbb{R}\oplus TN\\
&(t,x) &\mapsto &\frac{\partial}{\partial t}+X(t,x).
\end{array}
\end{equation}
The integral curves of $X$ of the form $(t,x(t))$ are particular solutions of the system
\begin{equation}\label{LieSystem}
\frac{{\rm d}x_i}{{\rm d}t}=X_i(t,x),\qquad i=1,\ldots,n.
\end{equation}
This is the so-called {\it associated system} with $X$.
Conversely, a system of this type describes the integral curves of the form $(t,x(t))$ of a unique $t$-dependent vector field on $N$,
namely,
\begin{equation}\label{TVF}
X(t,x)=\sum_{i=1}^nX_i(t,x)\frac{\partial}{\partial x_i}.
\end{equation}
This justifies to employ
the symbol $X$ to refer to both a $t$-dependent vector field and its associated
system.

\begin{definition} A {\it superposition rule} for a system $X$ on $N$ is a function $\Phi:N^{m}\times N\rightarrow N$ of the form $x=\Phi(x_{(1)}, \ldots,x_{(m)};k)$
satisfying that the general
 solution $x(t)$ of $X$ can be written as
  \begin{equation}\label{SupRule}
 x(t)=\Phi(x_{(1)}(t), \ldots,x_{(m)}(t);k),
  \end{equation}
where $x_{(1)}(t),\ldots,x_{(m)}(t)$ is any generic family of
particular solutions to $X$ and $N\ni k=(k_1\ldots,k_n)$ can be understood as a set of constants to be related
to initial conditions of $X$.
 \end{definition}

\begin{theorem}{\bf (The Lie--Scheffers Theorem)} A system $X$ on $N$ admits a superposition rule if and only if
\begin{equation}\label{LieDecom}
X(t,x)=\sum_{\alpha=1}^rb_\alpha (t)X_\alpha(x)
\end{equation}
for certain $t$-dependent functions
$b_1,\ldots,b_r$ and vector fields $X_1,\ldots,X_r$ on $N$ spanning an
$r$-dimensional real Lie algebra: a Vessiot--Guldberg Lie algebra of $X$.
\end{theorem}

If $X$ is a Lie system admitting a Vessiot--Guldberg Lie algebra isomorphic to a Lie algebra $\mathfrak{g}$, we call $X$ a $\mathfrak{g}$-Lie system \cite{Pi12}.

Lie systems and superposition rules can be generalized to PDEs as follows \cite{CGM07}. Consider a system of PDEs 
\begin{equation}\label{pde}
 \frac{\partial x_i}{\partial t_l}=X_{li}(t,x),\qquad i=1,\ldots,n,\qquad l=1,\ldots,s,
\end{equation}
where the $X_{li}:\mathbb{R}^s\times \mathbb{R}^n\mapsto \mathbb{R} \in X(t,x)$ are arbitrary functions.
Its particular solutions are sections $x:t\equiv (t_1,\ldots,t_s) \in\mathbb{R}^s\mapsto x(t)\in N$ of the bundle $(\mathbb{R}^s \times N,N,\pi_s:(t,x)\in \mathbb{R}^s\times N\mapsto x\in N)$.  In particular, we recover the simple case of ODEs when $s=1$.

Let us assume $t\in \mathbb{R}^s$. We call {\it $t$-dependent vector field} $X$ on $N$ a map $X:\mathbb{R}^s\times N\rightarrow {\rm T}N$ such that $\tau\circ X=\pi_s$.  We call $\bar{X}^l$, with $l=1,\ldots,s$, the {\it autonomization} of $X$ with respect to $t_l$, i.e.,
\begin{equation}\label{xx2}
 \bar{X}^l(t,x)\equiv\frac{\partial}{\partial t_l}+X(t,x).
\end{equation}

\begin{definition}
 A system of PDEs \eqref{pde} is a {\it PDE Lie system} if:
 \begin{itemize}
  \item  There exist vector fields $X_{1},\ldots,X_{r}$ on $N$ spanning an $r$-dimensional real Lie algebra $V$ satisfying that each $t$-dependent vector field $X_{l}(t,x)=\sum_{i=1}^n X_{li}(t,x){\partial}/{\partial x_i}$, where $l=1,\ldots,s$,
 can be written in the form
\begin{equation*}
 X_{l}(t,x)=\sum_{\alpha=1}^r{b_{{\alpha l}}(t)X_{\alpha}(x)},\qquad l=1,\dots,s
\end{equation*}
for certain $t$-dependent functions $b_{\alpha l}$.
\item Let $c_{\alpha\beta\gamma}$ be the structure constants of $X_1,\ldots,X_r$, i.e., $[X_\alpha,X_\beta]=\sum_{\gamma=1}^rc_{\alpha\beta\gamma}X_\gamma$. Then,
\begin{equation*}
 \frac{\partial b_{{\gamma} k}}{\partial t_l}(t)-\frac{\partial b_{{\gamma} l}}{\partial t_k}(t)+\sum_{\alpha,\beta=1}^rb_{{\alpha} l}(t)b_{{\beta} k}(t)c_{\alpha \beta \gamma}=0,\quad \quad l\neq k=1,\ldots,s,\quad \gamma=1,\ldots,r.
\end{equation*}
\end{itemize}
If these conditions are satisfied, we call $V$ a {\it Vessiot--Guldberg Lie algebra for the PDE Lie system (\ref{pde})}.
\end{definition}

\section{On certain Lie symmetries for Lie systems}

We are now concerned with the study of certain Lie symmetries for Lie systems. We prove that the features of these Lie symmetries are determined by the algebraic structure of a Vessiot--Guldberg Lie algebra of the Lie system. 

Consider again the Lie system $X$ given by (\ref{LieSystem}) admitting a Vessiot--Guldberg Lie algebra $V$ with basis $X_1,\ldots,X_r$ and structure constants $c_{\alpha\beta\gamma}$. Then, $X$ can be brought into the form (\ref{LieDecom})
for some $t$-dependent functions $b_1,\ldots,b_r$ (with $t\in\mathbb{R}$).
Let us study its Lie symmetries of the form
\begin{equation}\label{LieSym}
Y=f_0(t)\frac{\partial}{\partial t}+\sum_{\alpha=1}^rf_\alpha(t)X_\alpha,
\end{equation}
where $f_0,\ldots,f_r$ are certain $t$-dependent functions.  We denote by $\mathcal{S}^V_X$ the space of such Lie symmetries.
Recall that $Y\in\mathcal{S}^V_X$ if and only if
\begin{equation}\label{consym}
[Y,\bar X]=h \bar X
\end{equation}
for a function $h\in C^\infty(\mathbb{R}\times N)$, with $\bar X$ being the autonomization of $X$ given by (\ref{xx})  \cite{Olver}. From this, it follows that $\mathcal{S}^V_X$ is a real Lie algebra of vector fields.

Using properties of $\mathcal{S}_X^V$ above, we now characterize the elements of $\mathcal{S}^V_X$ as particular solutions of a family of Lie systems.

\begin{lemma}\label{SLS} The vector field $Y$ of the form (\ref{LieSym}) is a Lie symmetry for the Lie system (\ref{LieDecom}) if and only if the $t$-dependent functions $f_0,\ldots,f_r$ satisfy the system of differential equations
\begin{equation}\label{SysFun}
\frac {{\rm d}f_0}{{\rm d}t}=b_0(t),\quad \frac{{\rm d}f_\alpha}{{\rm d}t}=f_0\frac{{\rm d}b_\alpha}{{\rm d}t}(t)+b_\alpha(t)b_0(t) +\sum_{\beta,\gamma=1}^rb_\beta(t) f_\gamma c_{\gamma\beta\alpha},\qquad \alpha=1,\ldots,r,
\end{equation}
for a certain $t$-dependent function $b_0$.
\end{lemma}
\begin{proof}

From (\ref{LieDecom}), (\ref{LieSym}) and (\ref{consym}), we have that
$$
\begin{aligned}
\!\!\left[Y,\bar{X}\right]\!\!&=\!\!\left[f_0\frac{\partial}{\partial t}\!+\!\sum_{\alpha=1}^
{r}f_\alpha X_\alpha,\frac{\partial}{\partial
t}\!+\!\sum_{\beta=1}^rb_\beta X_\beta\right]\!\!\\
&=-\frac{{\rm d}f_0}{{\rm d}t}\frac{\partial}{\partial t}+\!\!\sum_{\alpha=1}^r\left(\left(f_0\frac{{\rm d} b_\alpha}{{\rm d} t}-\frac{{\rm d} f_\alpha}{{\rm d}
t}\right)X_\alpha+
\sum_{\beta=1}^rb_\beta f_\alpha[X_\alpha,X_\beta]\right)
\\
&=-\frac{{\rm d}f_0}{{\rm d}t}\frac{\partial}{\partial t}+\sum_{\alpha=1}^{r}\left(\left(f_0\frac{{\rm d} b_\alpha}{{\rm d} t}-\frac{{\rm d} f_\alpha}{{\rm d} t}\right)X_\alpha+\sum_{\beta,\gamma=1}^{r}b_{\beta} f_\alpha c_{
\alpha\beta\gamma}X_{\gamma}\right)\\
&=-\frac{{\rm d}f_0}{{\rm d}t}\frac{\partial}{\partial t}+\sum_{\alpha=1}^{r}\left(f_0\frac{{\rm d} b_\alpha}{{\rm d} t}-\frac{{\rm d} f_\alpha}{{\rm d}
t}+{\sum_{\beta,\gamma=1}^{r}{b_{\beta}f_\gamma c_{
\gamma\beta\alpha}}}\right)X_{\alpha}\\&=h \left(\frac{\partial }{\partial t}+\sum_{\alpha=1}^{r}b_{\alpha}X_{\alpha}\right).
\end{aligned}
$$
Thus, $[Y,\bar X]=h\bar X$ is equivalent to
$$\left(-\frac{{\rm d}f_0}{{\rm d}t}-h\right)\frac{\partial}{\partial t}+\sum_{\alpha=1}^{r}\left[\sum_{\beta,\gamma=1}^{r}{{b_{\beta}f_\gamma c_{
\gamma\beta\alpha}}}+f_0\frac{{\rm d} b_\alpha}{{\rm d} t}-\frac{{\rm d} f_\alpha}{{\rm d}
t}-hb_{\alpha}\right]X_{\alpha}=0.
$$
Since $\partial/\partial t, X_1,\ldots,X_r$, are linearly independent over $\mathbb{R}$, we obtain that $[Y,\bar X]=h\bar X$ if and only if (\ref{SysFun}) is fulfilled with $h$ being an arbitrary $t$-dependent function such that $b_0(t)=-h(t)$ for every $t\in\mathbb{R}$.
 Hence, the elements of $\mathcal{S}_X^V$ are the particular solutions of (\ref{SysFun}) for a $t$-dependent function $b_0$. 
\end{proof}

\begin{definition} We call (\ref{SysFun}) the {\it symmetry system} of the Lie system (\ref{LieDecom}) with respect to its Vessiot--Guldberg Lie algebra $V$. We write $\Gamma_X^V$ for the $t$-dependent vector field associated with (\ref{SysFun}), 
that is, 
\end{definition}
\begin{equation*}
 \Gamma_X^V=b_0(t)\frac{\partial}{\partial f_0}+\sum_{\alpha=1}^r\left(f_0\frac{{\rm d}b_{\alpha}}{{\rm d} t}(t)+b_0(t)b_{\alpha}(t)+\sum_{\gamma,\beta=1}^r b_{\beta}(t)f_{\gamma}c_{\gamma \beta \alpha}\right)\frac{\partial}{\partial f_{\alpha}}.
\end{equation*}

Note that (\ref{LieDecom}) may have different Vessiot--Guldberg Lie algebras (see \cite{Dissertations} for details). 
When $V$ is understood from context, we simply call (\ref{SysFun}) the symmetry system of $X$. Let us prove that (\ref{SysFun}) is a Lie system.

\begin{theorem}\label{The:SLS}
The system $\Gamma_{X}^V$ is a Lie system possessing a Vessiot--Guldberg Lie algebra
\begin{equation}\label{deVS}
(A_1\oplus_S A_2)\oplus_S V_L\simeq (\mathbb{R}^{r+1}\oplus_S\mathbb{R}^r)\oplus_S V/Z(V),
\end{equation}
where
$$
A_1=\langle Z_0,\ldots, Z_r\rangle \simeq \mathbb{R}^{r+1},\qquad A_2=\langle W_1,\ldots, W_r\rangle \simeq \mathbb{R}^r,\qquad V_L=\langle Y_1,\ldots,Y_r\rangle \simeq V/Z(V),
$$
with
\begin{equation*}\label{symfields}
Y_\alpha=\sum_{\beta,\gamma=1}^rf_\beta
c_{\beta\alpha \gamma}\frac{\partial}{\partial f_\gamma}, \qquad W_{\alpha}=f_0\frac{\partial}{\partial f_\alpha},\qquad Z_0=\frac{\partial}{\partial f_0},\qquad Z_\alpha=\frac{\partial}{\partial f_\alpha},\qquad \alpha=1,\ldots,r,
\end{equation*}
 we write $A \oplus_S B$ for the semi-direct sum of the ideal $A$ of $A+B$ with $B$, and $Z(V)$ is the center of the Lie algebra $V$. 
\end{theorem}
\begin{proof}
By defining $t$-dependent functions
$$
c_0\equiv b_0,\quad \bar b_0\equiv 0, \quad c_\alpha\equiv b_0b_\alpha, \quad \bar b_\alpha\equiv \frac{{\rm d}b_\alpha}{{\rm d}t},\qquad \alpha=1,\ldots,r,
$$
we can write
$$
\Gamma^V_X(t,f)=\sum_{\alpha=0}^r c_{\alpha}(t)Z_{\alpha}(f)+\sum_{\alpha=1}^r\left[\bar{b}_{\alpha}(t)W_{\alpha}(f)+b_{\alpha}(t)Y_{\alpha}(f)\right],\qquad f\in\mathbb{R}^{n+1}.
$$
Hence, $\Gamma^V_X$ is a $t$-dependent vector field taking values in the linear space $V^S\equiv A_1+A_2+V_L$. Let us show that $V^S$ is also a Lie algebra of vector fields. To do so, let us first prove that $V_L\simeq V/Z(V)$. Consider $Y_{\alpha}, Y_{\beta}\in V_L$. Recalling that $[X_{\alpha},X_{\beta}]=\sum_{\gamma=1}^rc_{\alpha\beta\gamma}X_{\gamma}$, we obtain
\begin{multline*}
\left[Y_{\alpha},Y_{\beta}\right]=\sum_{i,j,m,n=1}^{r}\left[c_{i \alpha j}f_{i}\frac{\partial}{
\partial f_{j}}, c_{m
\beta n}f_{m}\frac{\partial}{\partial f_{n}}\right]=
\sum_{i,j,m,n=1}^{r}c_{i \alpha j}c_{m
\beta n}\left(f_{i}\delta_{j}^{m}\frac{\partial}{
\partial x_{n}}-f_{m}\delta_{n}^{i}\frac{\partial}{\partial
f_{j}}\right)\\
=\sum_{i,m,n=1}^{r}c_{i\alpha m}c_{m \beta n}f_{i}
\frac{\partial}{\partial f_{n}}-\sum_{n,j,m=1}^{r}c_{n\alpha j}c_{m \beta n}f_{m
}\frac{\partial}{\partial f_j}=\sum_{i,m,n=1}^{r}(c_{i\alpha m}c_{m \beta n}+c_{m\alpha n}c_{\beta im})f_{i
}\frac{\partial}{\partial f_n}.
\end{multline*}
Using the Jacobi identity for the structure constants $c_{\alpha\beta\gamma}$, we see that
\begin{equation*}\label{sc}
\sum_{m,n=1}^{r}\left(c_{i\alpha m}c_{m \beta n}+c_{\alpha \beta m}c_{m i n}+c_{\beta i m}c_{m \alpha n}\right)=0,\qquad \forall i,\alpha,\beta=1,\ldots,r.
\end{equation*}
From this, 
\begin{equation*}
\left[Y_{\alpha},Y_{\beta}\right]=\sum_{i,m,n=1}^{r}(c_{i\alpha m}c_{m \beta n}+c_{m\alpha n}c_{\beta im})f_{i
}\frac{\partial}{\partial f_n}=-\sum_{i,m,n=1}^{r}c_{\alpha \beta m}c_{m i n}f_{i
}\frac{\partial}{\partial f_n}=\sum_{m=1}^rc_{\alpha \beta m}Y_{m}.
\end{equation*}
So, $Y_1,\ldots,Y_{r}$ span a Lie algebra. We can define a Lie algebra morphism $\phi:V\rightarrow V_L$ of the form $\phi(X_\alpha)=Y_\alpha$ for 
$\alpha=1,\ldots,r$. 
The vector fields $Y_1,\ldots,Y_r$ do not need to be linearly independent. Let us show this. We can assume with no loss of generality that $X_1,\ldots,X_s$, with $s\leq r$, form a basis for $\ker \phi$. Since $\phi(X_\alpha)=0$ for $\alpha=1,
\ldots,s$, we have $f_\beta c_{\beta\alpha\gamma}=0$ for $\alpha=1,\ldots,s$ and $\beta,\gamma=1,\ldots,r$. Thus, we see that $[X_\alpha,X_\beta]=0$ for $\alpha=1,\ldots,s$ and $\beta=1,\ldots,r$. This means that $X_\alpha\in Z(V)$. Conversely,  we get by similar arguments that if $X\in Z(V)$, then $X\in \ker \phi.$ Hence, $X\in \ker\phi$ if and only if $X\in Z(V)$. In consequence, $\ker \phi=Z(V)$ and $Y_1,\ldots,Y_r$ span a Lie algebra isomorphic to $V/Z(V)$. 

It is obvious that $A_1$ is an ideal of $A_1+A_2$. Moreover, as $[A_1,V_L]\subset A_1$ and $[A_2,V_L]\subset A_2$, then $A_1\oplus_{S} A_2$ is an ideal of $V^S$. Consequently, $V^S$ is a Lie algebra of the form (\ref{deVS}).
\end{proof}

\begin{definition} We say that the {\it Lie systems $X_1$ and $X_2$ are isomorphic} when they take values in two isomorphic Vessiot--Guldberg Lie algebras $V_1, V_2$ and there exists a Lie algebra isomorphism $\phi:V_1\rightarrow V_2$ such that $(X_2)_t=\phi((X_1)_t)$ for each $t\in\mathbb{R}$.

\end{definition}

\begin{proposition}\label{SymTheo} Given two isomorphic Lie systems $X_1$ and $X_2$ related to Vessiot--Guldberg Lie algebras $V_1, V_2$, their symmetry systems relative to such Lie algebras are, up to a change of basis in $V_1$ and/or $V_2$, the same.
\end{proposition}

\section{Lie algebras of Lie symmetries for Lie systems}
In this section we study different Lie subalgebras of $\mathcal{S}^V_X$. Their interest resides in the fact that, when finite-dimensional, they can be integrated to form Lie group actions of symmetries for $X$ that can be employed to simplify the Lie system they are referred to.

\begin{lemma}
The space of functions $C^{\infty}(\mathbb{R})$ can be endowed with a Lie bracket given by
\begin{equation}\label{LSbracket}
 \{f,\bar f\}_{\mathbb{R}}=f\frac{{\rm d} \bar f}{{\rm d}t}-\bar f\,\frac{{\rm d} f}{{\rm d} t},\qquad \forall f,\bar f\in C^\infty(\mathbb{R}).
\end{equation}
\end{lemma}
\begin{proof}
In order to prove that \eqref{LSbracket} is a Lie bracket, we must show that (\ref{LSbracket}) is  bilinear, antisymmetric and satisfies the Jacobi identity.
From its definition (\ref{LSbracket}) is clearly bilinear and antisymmetric. To see that (\ref{LSbracket}) holds the Jacobi identity, we consider the map $\phi:f\in C^{\infty}(\mathbb{R}) \mapsto f\partial/\partial t\in \Gamma({\rm T}\mathbb{R})$. Observe that $\phi$ is a linear isomorphism. 

Moreover, 
$$
\phi(\{f,\bar f\}_\mathbb{R})=\{f,\bar f\}_\mathbb{R}\frac{\partial}{\partial t}=\left(f\frac{{\rm d} \bar f}{{\rm d}t}-\bar f\,\frac{{\rm d}f}{{\rm d}t}\right)\frac{\partial}{\partial t}=\left[f\frac{\partial}{\partial t},\bar f\frac{\partial}{\partial t}\right]=[\phi(f),\phi(\bar f)],\,\, \forall f,\bar f\!\in\! C^\infty(\mathbb{R}).
$$
By using the Jacobi identity for vector fields on $\mathbb{R}$ with respect to the Lie bracket $[\cdot,\cdot]$, we obtain
\begin{multline*}
 \phi (\{\{f,\bar f\}_\mathbb{R},\bar{\bar f}\}_\mathbb{R}+\{\{\bar f,\bar{\bar{f}}\}_\mathbb{R},f\}_\mathbb{R}+\{\{\bar{\bar f},f\}_\mathbb{R},\bar{f}\}_\mathbb{R})=\\
[[\phi(f),\phi(\bar{f})],\phi(\bar{\bar f})]+[[\phi(\bar f),\phi(\bar{\bar f})],\phi(f)]+[[\phi(\bar{\bar f}),\phi(f)],\phi(\bar{f})]=0,\,\, \forall f,\bar f,\bar{\bar f}\in C^\infty(\mathbb{R}).
\end{multline*}
Since $\phi$ is a linear isomorphism, then
$$
\{\{f,\bar f\}_\mathbb{R},\bar{\bar f}\}_\mathbb{R}+\{\{\bar f,\bar{\bar f}\}_\mathbb{R},f\}_\mathbb{R}+\{\{\bar{\bar f},f\}_\mathbb{R},\bar f\}_\mathbb{R}=0
$$
and (\ref{LSbracket}) satisfies the Jacobi identity giving rise to a Lie bracket on $C^\infty(\mathbb{R})$. Moreover, $\phi$ becomes a Lie algebra isomorphism.
\end{proof}
\begin{definition}
Let $X$ be a Lie system on $N$ with a Vessiot--Guldberg Lie algebra $V$ and let $\mathfrak{W}$ be a nonempty set of
$t$-dependent functions that form a Lie algebra with respect to the Lie bracket defined in \eqref{LSbracket}. We call $\mathcal{S}_{X,\mathfrak{W}}^{V}$ the space
\begin{equation*}
 \mathcal{S}_{X,\mathfrak{W}}^{V}=\left\{Y\in \mathcal{S}_X^{V} \mid  Y=f_0\frac{\partial}{\partial t}+\sum_{\alpha=1}^{r}f_\alpha X_{\alpha},\,\, f_0\in \mathfrak{W}\right\},
\end{equation*}
where $X_1,\ldots,X_r$ is a basis for $V$.
\end{definition}

\begin{proposition}\label{finite}
The space of symmetries $\mathcal{S}_{X,\mathfrak{W}}^{V}$ is a Lie algebra of symmetries of $X.$
\begin{proof}
Since $\mathfrak{W}$ and $\mathcal{S}_X^V$ are linear spaces, the linear combinations of elements of $\mathcal{S}^V_{X,\mathfrak{W}}$ belong to $\mathcal{S}^V_{X,\mathfrak{W}}$. So, this space becomes a vector space. Moreover, given two elements, $Y,Y^*\in \mathcal{S}_{X,\mathfrak{W}}^{V}$, their Lie bracket reads
\begin{multline*}
 [Y,Y^*]=\left[f_0(t)\frac{\partial}{ \partial t}+\sum_{\alpha=1}^{r}f_\alpha(t) X_{\alpha},f_0^*(t)\frac{\partial}{\partial t}+\sum_{\beta=1}^{r}{f}^*_\beta(t) X_{\beta}\right]=\\
 \{f_0,{f^*_0}\}_{\mathbb{R}}\frac{\partial}{ \partial t}+\sum_{\beta=1}^{r}\left[\left(f_0\frac{{\rm d}{f^*_{\beta}}}{{\rm d}t}-{f^*_0}\frac{df_{\beta}}{d t}\right)X_{\beta}+\sum_{\alpha,\gamma=1}^{r}f_{\alpha}{f^*_{\beta}}c_{\alpha\beta\gamma}X_{\gamma}\right]=\\
\{f_0,{f^*_0}\}_{\mathbb{R}}\frac{\partial}{\partial t}+\sum_{\gamma=1}^{r}\left[\left(f_0\frac{{\rm d}{f^*_{\gamma}}}{{\rm d} t}-{f^*_0}\frac{{\rm d} f_{\gamma}}{{\rm d} t}\right)+\sum_{\alpha,\beta=1}^{r}f_{\alpha}{f^*_{\beta}}c_{\alpha\beta\gamma}\right]X_{\gamma}.
\end{multline*}
Since $\mathcal{S}_X^V$ is a Lie algebra and $Y,Y^*\in\mathcal{S}^V_{X}$, then $[Y,Y^*]\in \mathcal{S}_X^V$. As  additionally $\{f_0,f^*_0\}_\mathbb{R}\in \mathfrak{W}$, then $[Y,Y^*]\in \mathcal{S}^V_{X,\mathfrak{W}}$. Hence, the Lie bracket of elements of $\mathcal{S}^V_{X,\mathfrak{W}}$ belongs to $\mathcal{S}^V_{X,\mathfrak{W}}$, which becomes a Lie algebra.
\end{proof}
\end{proposition}

\begin{corollary}
Given a Lie system $X$ on $N$ related to a Vessiot--Guldberg Lie algebra $V,$ the elements of $\mathcal{S}^{V}_{X,\mathfrak{W}}$
with
\begin{enumerate}
 \item $\mathfrak{W}=\{f_0\in C^\infty(\mathbb{R})\,\,|\,\,{\rm d}f_0/{\rm d}t=0\}$,
 \item $\mathfrak{W}=\{f_0\in C^\infty(\mathbb{R})\,\,|\,\,f_0=0\}$,
\end{enumerate}
are finite-dimensional Lie algebras of vector fields. In the second case, $S^V_{X,\mathfrak{W}}$ is isomorphic to $V.$
\end{corollary}

\begin{proof} In both cases, $\mathfrak{W}$ is nonempty. In the first case, the functions with ${\rm d}f_0/{\rm d}t=0$ are constant. These functions form an Abelian Lie algebra with
respect to the Lie bracket $\{\cdot,\cdot\}_\mathbb{R}$. In view of Proposition \ref{finite}, the space $\mathcal{S}^V_{X,\mathfrak{W}}$ is a Lie algebra.

In the second case, the function zero
is also a zero-dimensional Lie algebra relative to $\{\cdot,\cdot\}_\mathbb{R}$. Since $f_0=0$ and using Proposition \ref{finite}, we obtain that $\mathcal{S}^V_{X,\mathfrak{W}}$ can be understood as a Lie algebra of $t$-dependent vector fields taking values in $V$.
To prove that $\mathcal{S}^V_{X,\mathfrak{W}}\simeq V$, let us consider the morphism which maps each $t$-dependent vector field with its value at $t=0$, namely
$$\begin{array}{rccc}
\phi:&\mathcal{S}^V_{X,\mathfrak{W}}&\longrightarrow &V\\
    &Z&\mapsto &Z_0.
\end{array}
$$
Let $X_1,\ldots,X_r$ be a basis for $V$. From Lemma (\ref{SLS}), we have that $(f_1(t),\ldots,f_r(t))$ is a particular solution of the system
\begin{equation*}
\frac{ {\rm d}f_\alpha}{{\rm d}t}=\sum_{\delta,\beta=1}^rb_\beta(t) f_\delta
c_{\delta \beta \alpha},\qquad \alpha=1,\ldots,r.
\end{equation*}
For each initial condition $f_\alpha(0)=c_\alpha \in \mathbb{R}$, with $\alpha=1,\ldots,r$, i.e., by fixing $Z_0$, there exists a unique solution of the above system. Hence, there exists a unique $t$-dependent vector field $Z$ of $\mathcal{S}^V_{X,\mathfrak{W}}$
with $Z_0=\sum_{\alpha=0}^r c_\alpha X_\alpha$. Thus, $\phi$ is a bijection. Using that for two vector fields $Z_1,Z_2\in \mathcal{S}^V_{X,\mathfrak{W}}$ we have $[Z_1,Z_2]\in \mathcal{S}^V_{X,\mathfrak{W}}$ and $[Z_1,Z_2]_t=[(Z_1)_t,(Z_2)_t]$, we see that
$\phi$ is a Lie algebra morphism and $\mathcal{S}^V_{X,\mathfrak{W}}\simeq V$.
\end{proof}

\section{Applications to systems of ODEs and HODEs}

Let us work out the symmetry systems and related Lie symmetries for some Lie systems of interest. In particular, we will illustrate that Proposition
 \ref{SymTheo} enables us to determine simultaneously Lie symmetries for different Lie systems with isomorphic Vessiot--Guldberg Lie algebras.
\subsection{Lie symmetries for sl(2,$\mathbb{R}$)-Lie systems}\label{LSSL2}

Let us obtain the symmetry systems and related Lie symmetries for $\mathfrak{sl}(2,\mathbb{R})$-Lie systems. This shall be used in following subsections to obtain simultaneously Lie symmetries of isomorphic $\mathfrak{sl}(2,\mathbb{R})$-Lie systems appearing in the physics and/or mathematical literature.

Let us choose a basis of vector fields $\{X_1,X_2,X_3 \}$ of $V\simeq \mathfrak{sl}(2,\mathbb{R})$ with commutation relations
\begin{equation}\label{ComSl2}
[X_1,X_2]=X_1,\quad [X_1,X_3]=2X_2,\quad [X_2,X_3]=X_3.
\end{equation}
Every Lie system with Vessiot--Guldberg Lie algebra $V$ can be brought into the form
\begin{equation}\label{sl2}
X=b_1(t)X_1+b_2(t)X_2+b_3(t)X_3
\end{equation}
for certain $t$-dependent functions $b_1$, $b_2$ and $b_3$.
The Lie symmetries of $\mathcal{S}_X^V$ take the form
$$Y=f_0(t)\frac{\partial}{\partial t}+f_1(t)X_1+f_2(t)X_2+f_3(t)X_3,$$
where $f_0,f_1,f_2,f_3$ are some $t$-dependent functions to be determined. In view of (\ref{SysFun}) and the commutation relations (\ref{ComSl2}), the symmetry system for $X$ relative to $V$ reads
\begin{equation}\label{slsys}
\left\{\begin{aligned}
\frac{{\rm d}f_0}{{\rm d}t}&=b_0(t),\\
 \frac{{\rm d}f_1}{{\rm d}t}&=f_0\frac{{\rm d} b_1}{{\rm d} t}(t)+f_1b_2(t)-f_2b_1(t)+b_0(t)b_1(t),\\
 \frac{{\rm d}f_2}{{\rm d}t}&=f_0\frac{{\rm d} b_2}{{\rm d} t}(t)+2f_1b_3(t)-2f_3b_1(t)+b_0(t)b_2(t),\\
 \frac{{\rm d}f_3}{{\rm d}t}&=f_0\frac{{\rm d} b_3}{{\rm d} t}(t)+f_2b_3(t)-f_3b_2(t)+b_0(t)b_3(t).\\
\end{aligned}\right.
\end{equation}

As stated in Theorem \ref{The:SLS}, this is a Lie system. Indeed, system \eqref{slsys} is related to the $t$-dependent vector field
\begin{multline*}
\Gamma^{\mathfrak{sl}(2,\mathbb{R})}_X=\frac{{\rm d}b_1(t)}{{\rm d}t}W_1+\frac{{\rm d}b_2(t)}{{\rm d}t}W_2+\frac{{\rm d}b_3(t)}{{\rm d}t}W_3+b_0(t)Z_0+b_0(t)b_1(t)Z_1+b_0(t)b_2(t)Z_2\\\qquad\qquad\qquad+b_0(t)b_3(t)Z_3+b_1(t)Y_1+b_2(t)Y_2+b_3(t)Y_3,
\end{multline*}
where
\begin{equation}\label{com1}
 Z_\alpha=\frac{\partial}{\partial f_\alpha},\qquad \alpha=0,1,2,3,\qquad W_\beta=f_0\frac{\partial}{\partial f_\beta},\qquad \beta=1,2,3
\end{equation}
and
\begin{equation}\label{com3}
 Y_1=-f_2\frac{\partial}{\partial f_1}-2f_3\frac{\partial}{\partial f_2},\quad Y_2=f_1\frac{\partial}{\partial f_1}-f_3\frac{\partial}{\partial f_3},\quad Y_3=2f_1\frac{\partial}{\partial f_2}+f_2\frac{\partial}{\partial f_3}.
\end{equation}
These vector fields hold
\begin{equation}\label{comm0}
 [Y_1,Y_2]=Y_1,\quad [Y_1,Y_3]=2Y_2,\quad [Y_2,Y_3]=Y_3.
\end{equation}
Since $Z(V)=\{0\}$, then $V_L=\langle Y_1,Y_2,Y_3\rangle$ is a Lie algebra isomorphic to $V/Z(V)\simeq\mathfrak{sl}(2,\mathbb{R})$ as stated in Theorem \ref{The:SLS}.
The rest of commutation relations read
\begin{equation*}
\begin{array}{lllll}
&[Y_1,Z_0]=0,\quad &[Y_1,Z_1]=0, &[Y_1,Z_2]=Z_1, &[Y_1,Z_3]=2Z_2,\\
&[Y_2,Z_0]=0,\quad &[Y_2,Z_1]=-Z_1, &[Y_2,Z_2]=0, &[Y_2,Z_3]=Z_3,\\
&[Y_3,Z_0]=0,\quad &[Y_3,Z_1]=-2Z_2, &[Y_3,Z_2]=-Z_3, &[Y_3,Z_3]=0.
\end{array}
\end{equation*}
Moreover,
\begin{equation*}
\begin{array}{llll}
&[Y_1,W_1]=0,\quad &[Y_1,W_2]=W_1,\quad &[Y_1,W_3]=2W_2\\
&[Y_2,W_1]=-W_1,\quad &[Y_2,W_2]=0,\quad &[Y_2,W_3]=W_3,\\
&[Y_3,W_1]=-2W_2,\quad &[Y_3,W_2]=-W_3,\quad &[Y_3,W_3]=0
\end{array}
\end{equation*}
and
\begin{equation}\label{comm3}
\begin{gathered}
\left[Z_0,W_j\right]=Z_j,\quad [Z_i,W_j]=0,\quad  [W_i,W_j]=0,\quad i,j=1,2,3,\\ 
[Z_\alpha,Z_\beta]=0,\quad \alpha,\beta=0,\ldots,3.
\end{gathered}
\end{equation}
Hence, $A_1=\langle Z_0,Z_1,Z_2,Z_3 \rangle$ is an ideal of $A_1+A_2$. And $A_1+A_2$ is an ideal
of $A_1+A_2+V_L$, with $A_2=\langle W_1,W_2,W_3 \rangle$.

\subsubsection{A first-order Riccati equation}

The previous results can be applied in studying the first-order Riccati equation \cite{Ince,Riccati}
\begin{equation}\label{Ricc}
\frac{{\rm d}x}{{\rm d}t}=\eta(t)+x^2,
\end{equation}
where $\eta$ is an arbitrary $t$-dependent function.

It is well known that the Riccati equation is a Lie system with a Vessiot--Guldberg Lie algebra isomorphic to $V^{\rm Ric}\simeq \mathfrak{sl}(2,\mathbb{R})$ \cite{CLS13}.
Indeed, equation (\ref{Ricc}) has the associated $t$-dependent vector field $X^{\rm Ric}=X^{\rm Ric}_3+\eta(t)X^{\rm Ric}_1,$ where
\begin{equation}\label{VecRic}
X^{\rm Ric}_1=\frac{\partial}{\partial x},\quad X^{\rm Ric}_2=x\frac{\partial}{\partial x} \quad X^{\rm Ric}_3=x^2\frac{\partial}{\partial x}
\end{equation}
satisfy the same commutation relations as $X_1,X_2,X_3$ in (\ref{ComSl2}). In view of this, (\ref{Ricc}) is related to a $t$-dependent vector field taking values in a finite-dimensional Lie algebra of vector fields isomorphic to $\mathfrak{sl}(2,\mathbb{R})$. Then, (\ref{Ricc}) is an $\mathfrak{sl}(2,\mathbb{R})$-Lie system. Moreover, we can consider $X^{\rm Ric}$ as a particular case of system (\ref{sl2}). By applying the results of the previous sections to generic $\mathfrak{sl}(2,\mathbb{R})$-Lie systems, we find that the symmetry system for $X^{\rm Ric}$ is of the form (\ref{slsys})  with $b_1=\eta, b_2=0$ and $b_3=1$, namely
\begin{equation}\label{ric2}
\left\{\begin{aligned}
\frac{{\rm d}f_0}{{\rm d}t}&=b_0(t),\\
 \frac{{\rm d}f_1}{{\rm d}t}&=f_0\frac{{\rm d}\eta}{{\rm d}t}(t)-\eta(t) f_2+b_0(t)\eta(t),\\
 \frac{{\rm d}f_2}{{\rm d}t}&= 2f_1-2\eta(t) f_3,\\
 \frac{{\rm d}f_3}{{\rm d}t}&=f_2+b_0(t).
\end{aligned}\right.
\end{equation}

We can recover and generalize the results given in \cite{GS10} by means of our approach. From their expressions in (\ref{ric2}), we can differentiate ${\rm d}f_3/{\rm d}t$ twice and ${\rm d}f_2/{\rm d}t$ once. By substituting ${\rm d}^2f_2/{\rm d}t^2$ in ${\rm d}^3f_3/{\rm d}t^3$ and using the remaining equations in (\ref{ric2}), we obtain that	
$$
\frac{{\rm d}^3f_3}{{\rm d}t^3}=\frac{{\rm d}^3f_0}{{\rm d}t^3}+2\frac{{\rm d}f_1}{{\rm d}t}-2\frac{{\rm d}\eta}{{\rm d}t}f_3-2\eta(t)\frac{{\rm d}f_3}{{\rm d}t}.
$$
By substituting the value of ${{\rm d}f_1}/{{\rm d}t}$ from (\ref{ric2}) and using that $f_2={\rm d}f_3/{\rm d}t-b_0(t)$, we obtain the following equation for $f_3$ in terms of the coefficients of \eqref{Ricc} and $f_0$:
\begin{equation}\label{rCS}
\frac{{\rm d}^3f_3}{{\rm d}t^3}=\frac{{\rm d}^3f_0}{{\rm d}t^3}+4b_0\eta(t)+2\frac{{\rm d}\eta}{{\rm d}t}f_0-2\frac{{\rm d}\eta}{{\rm d}t}f_3-4\eta(t)\frac{{\rm d}f_3}{{\rm d}t}.
\end{equation}
From this, we can retrieve the following corollary given in \cite{GS10}.

\begin{corollary} The Riccati equation (\ref{Ricc}) admits the Lie symmetry
$$
Y=f_0\frac{\partial}{\partial t}-\frac 12 \frac{{\rm d}^2f_0}{{\rm d}t^2}\frac{\partial}{\partial x}-\frac{{\rm d}f_0}{{\rm d}t}x\frac{\partial}{\partial x},
$$
where
\begin{equation}\label{conRicEsp}
\frac{{\rm d}^3f_0}{{\rm d}t^3}+4b_0(t)\eta(t)+2\frac{{\rm d}\eta}{{\rm d}t}f_0=0.
\end{equation}

\end{corollary}

\begin{proof}
Since we are looking for Lie symmetries with $f_3=0$, equation (\ref{rCS}) reduces to (\ref{conRicEsp}). Moreover, by substituting $f_3=0$ in (\ref{ric2}), we obtain that
$$
\frac{{\rm d}f_3}{{\rm d}t}=b_0(t)+f_2=0,\qquad \frac{{\rm d}f_2}{{\rm d}t}=2f_1,
$$
which yields $f_2=-{{\rm d}f_0}/{{\rm d}t}$ and $2f_1=-{{\rm d}^2f_0}/{{\rm d}t^2}$. Hence, the corollary follows.
\end{proof}

Going back to general symmetries of (\ref{Ricc}), we can obtain some of its Lie symmetries by solving (\ref{ric2}) for certain values of $\eta(t)$ summarized in Table \ref{table1}.
\begin{table}[t] {\footnotesize
 \noindent
\caption{{\small Lie symmetries for Riccati equations (\ref{Ricc}) for values of the function $\eta(t)$.} We assume $f_0=k\in\mathbb{R}$. Function $f_3$ can be derived from (\ref{conRicEsp}). Function $f_2$ can be obtained from $f_3$, as
$f_2={\rm d}f_3/{{\rm d}t}$. Function $f_1$ reads 
$f_1=k\eta(t)-\int{\eta(t) f_3{\rm d}t}.$ In the following $Airy_A$ and $Airy_B$ denote the Airy and Bairy functions and $J,Y$ are the Bessel functions of first and second kind.
}
\label{table1}
\medskip
\noindent\hfill
 \begin{tabular}{  p{1.4cm}    p{14.5cm} l}
\hline
\hline
 &\\[-1.9ex]
$\eta$ & $f_3(t)$ \\[+1.0ex]
\hline
 &\\[-1.9ex]
$\dfrac{k}{at+b}$& $\displaystyle{k+c_1(at+b)J^2\left(1,2\sqrt{(at+b)/a^2}\right)+c_2(at+b)Y^2\left(1,2\sqrt{{at+b}/{a^2}}\right)}$\\[+1.0ex]
&$\qquad\qquad\qquad \qquad\qquad\qquad\qquad \qquad\qquad\qquad +\,c_3(at+b)J\left(1,2\sqrt{t+b/{a}}\right)Y\left(1,2\sqrt{t+b/{a}}\right)$ \\[+1.0ex]
$\dfrac{k}{(at+b)^2}$& $\displaystyle{-\frac{kat}{b}+c_1\left(\frac{at+b}{a}\right)^{\frac{a+\sqrt{a^2-4k}}{a}}+c_2\left(\frac{at+b}{a}\right)^{-\frac{-a+\sqrt{a^2-4k}}{a}}+c_3(at+b),}$\\[+1.0ex]
$\displaystyle{at+b}$ &$\displaystyle{k+c_1{\rm Airy}_A\left(-\frac{at+b}{a^{2/3}}\right)^2+c_2{\rm Airy}_B\left(-\frac{at+b}{a^{2/3}}\right)^2+c_3{\rm Airy}_A\left(-\frac{at+b}{a^{2/3}}\right){\rm Airy}_B\left(-\frac{at+b}{a^{2/3}}\right)}$
\\[+1.0ex]
\hline
\hline
 \end{tabular}
\hfill}
\end{table}

\subsubsection{Cayley--Klein Riccati equation}

Let us generalize the results of the previous section to a more general family of Riccati equations over different $\mathbb{R}$-linear algebras of numbers, namely real vector spaces along with a certain bilinear product.
In what follows, we summarize roughly some results given in~\cite{Yaglom} (see from page 265).

Consider the real plane with coordinates $\{x,y\}$ and an ``additional" unit $\iota$ such that
\begin{equation}
\label{ba}
\iota^2 \in\{ -1,+1,0\}    .
\end{equation}
Next, we define the real vector space of elements $z$ of the form
\begin{equation}
\label{bb}
 z:\equiv x+{\bf \iota} y,
\end{equation}
with the standard addition and multiplication: $(x_1+\iota y_1) +(x_2+\iota y_2)=x_1+x_2+\iota (y_1+y_2)$ and $\lambda(x_1+\iota y_1):=\lambda x_1+\iota \lambda y_1$ for arbitrary $x_1,x_2,y_1,y_2,\lambda \in\mathbb{R}$. We define $\bar z\equiv x-\iota y$ and a bilinear multiplication in the usual way: $(x_1+\iota y_1)(x_2+\iota y_2)=x_1x_2+\iota^2y_1y_2+\iota(x_1y_2+y_1x_2)$. Hence,
\begin{equation}
\label{bc}
 |z|^2\equiv  z\bar z= x^2-\iota^2 y^2,\qquad z^2=x^2+\iota^2 y^2 +2 \iota x y.
\end{equation}
In this way, we obtain three different algebra structures on $\mathbb{R}^2$ over $\mathbb{R}$, namely:

 \begin{itemize}

 \item Case $\iota^2=-1$. We say that $\iota$ is an {\it elliptical} additional unit. Here we are dealing with the usual complex numbers: $\iota\equiv {\rm i}$ and $z\in\mathbb C$. Hence,
\begin{equation}
 \label{bca}
   |z|^2= z\bar z= x^2+ y^2,\qquad z^2=x^2- y^2 +2 {\rm i} x y,\qquad z\in\mathbb C.
\end{equation}

  \item Case $\iota^2=+1$. We call $\iota$ a {\it hyperbolic} additional unit and we are dealing with the  so-called {\it split complex numbers} $z\in\mathbb C^\prime$. The additional unit is usually known as {\it double} or {\it Clifford} unit $\iota\equiv {\rm e}$~\cite{Yaglom}. Thus,
 \begin{equation}
 \label{bcb}
   |z|^2= z\bar z= x^2- y^2,\qquad z^2=x^2+y^2 +2 {\rm e} x y,\qquad z\in \mathbb C^\prime.
\end{equation}

  \item Case $\iota^2=0$. This last possibility corresponds to a {\it parabolic} additional unit, $\iota\equiv \varepsilon$, and $z$ is known as a {\it dual} or {\it Study} number $z\in\mathbb D$ ~\cite{Yaglom}  which can be regarded as a {\it contracted} case since
   \begin{equation}
 \label{bcc}
   |z|^2= z\bar z= x^2 ,\qquad z^2=x^2+ 2 \varepsilon x y,\qquad z\in \mathbb D.
\end{equation}

  \end{itemize}

We define the {\it Cayley--Klein  Riccati equation} to be the generalization of the standard Riccati equation to $z\in \{ \mathbb C,\mathbb C^\prime, \mathbb D\}$:
\begin{equation}\label{CKRE}
\frac{{\rm d} z}{{\rm d} t}=b_1(t)+b_2(t)z+b_3(t)z^2,\qquad z= x+\iota y.
\end{equation}
For real $t$-dependent coefficients $b_1$, $b_2$ and $b_3$, this equation gives rise to a system $X^{\rm CK}$ of two real differential equations
\begin{equation}\label{CKRE2}
\frac{{\rm d} x}{{\rm d} t}=b_1(t)+b_2(t)x+b_3(t)(x^2+\iota ^2y^2),\qquad \frac{{\rm d} y}{{\rm d} t}=b_2(t)y+2b_3(t)xy.
\end{equation}

The system (\ref{CKRE2}) can be described by means of the $t$-dependent vector field
\begin{equation*}
X^{\rm CK}=b_1(t)X^{\rm CK}_1+b_2(t)X^{\rm CK}_2+b_3(t)X^{\rm CK}_3,
\end{equation*}
where 
\begin{equation*}
X^{\rm CK}_1=\frac{\partial}{\partial x},\quad X^{\rm CK}_2=x\frac{\partial}{\partial x}+y\frac{\partial}{\partial y},\quad X_3^{\rm CK}=(x^2+\iota ^2 y^2)\frac{\partial}{\partial x}+2xy\frac{\partial}{\partial y}
\end{equation*}
satisfy the same commutation relations  as $X_1,X_2,X_3$ in (\ref{ComSl2}). Since $X_1,X_2,X_3$ span a Lie algebra $V\simeq \mathfrak{sl}(2,\mathbb{R})$, the vector fields $X^{\rm CK}_1,X^{\rm CK}_2$ and $X^{\rm CK}_3$ span a real Lie algebra, $V^{\rm CK}$, which is isomorphic to the Lie algebra $\mathfrak{sl}(2,\mathbb{R})$ independently of the square of $\iota$. Then, \eqref{CKRE2} is an $\mathfrak{sl}(2,\mathbb{R})$-Lie system. 

Moreover, we can define a the Lie algebra isomorphism $\phi:V^{\rm CK}\rightarrow V^{\rm }$ mapping
$$
\phi(X_\alpha^{\rm CK})=X_{\alpha},\qquad \alpha=1,2,3,
$$
where $X_\alpha$, with $\alpha=1,2,3$, are the generic vector fields in (\ref{ComSl2}).
Hence, we have that $\phi (X_t^{\rm CK})=X_t$ for every $t\in \mathbb{R}$ with $X$ being the generic $\mathfrak{sl}(2,\mathbb{R})$-Lie system (\ref{sl2}). Then, (\ref{CKRE2}) is isomorphic to  (\ref{sl2}) and, in view of Proposition \ref{SymTheo}, the symmetry system of $X^{\rm CK}$ becomes (\ref{slsys}).

\subsubsection{Quaternionic Riccati equation with t-dependent real coefficients}
Quaternions are the elements of the real vector space $\mathbb{H}\simeq\mathbb{R}^4$ of elements of the form $q=q_0+q_1{\rm i}+q_2{\rm j}+q_3{\rm k}$, where $q_0,q_1,q_2,q_3\in \mathbb{R}$ and with the standard sum of elements and multiplication by real numbers. We can define a multiplication of quaternions (see \cite{CS03} for details) by assuming that real numbers commute with all quaternions and
that the following operations are fulfilled
$$
\begin{gathered}
{\rm i}^2=-1,\quad {\rm j}^2=-1,\quad {\rm k}^2=-1,\\
{\rm i\cdot j}=-{\rm j\cdot i}={\rm k},\quad {\rm k\cdot i}=-{\rm i\cdot k}={\rm j},\quad {\rm j\cdot k}=-{\rm k\cdot j}={\rm i}.
\end{gathered}
$$

The {\it quaternionic Riccati equation}  \cite{W714} takes the form
\begin{equation}\label{qric}
\frac{{\rm d}q}{{\rm d}t}=b_1(t)+a_1(t)q+qa_2(t)+qb_3(t)q,
\end{equation}
where $q$ and the $t$-dependent functions $a_1,a_2,b_1,b_3:\mathbb{R}\rightarrow \mathbb{H}$ take values in $\mathbb{H}$ \cite{CS03}. The existence of periodic solutions for particular cases of (\ref{qric}) has been studied in \cite{CM06,W714} and, for real $t$-dependent coefficients, in \cite{W714}. In this work, we focus on the latter case.

Writing $q$ in coordinates, we obtain that \eqref{qric} reads
\begin{equation}
\left\{\begin{aligned}\label{qLS}
\frac{{\rm d}q_0}{{\rm d}t}&=b_1(t)+b_2(t)q_0+b_3(t)(q_0^2-q_1^2-q_2^2-q_3^2),\\
\frac{{\rm d}q_1}{{\rm d}t}&=b_2(t)q_1+2b_3(t)q_0q_1,\\
\frac{{\rm d}q_2}{{\rm d}t}&=b_2(t)q_2+2b_3(t)q_0q_2,\\
\frac{{\rm d}q_3}{{\rm d}t}&=b_2(t)q_3+2b_3(t)q_0q_3,
\end{aligned}\right.
\end{equation}
where $b_2(t)\equiv a_1(t)+a_2(t).$ This system is related to the $t$-dependent vector field
\begin{equation}
X^{\mathbb{H}}=b_1(t)X^{\mathbb{H}}_1+b_2(t)X^{\mathbb{H}}_2+b_3(t)X^{\mathbb{H}}_3,
\end{equation}
with$$
\begin{gathered}
X^{\mathbb{H}}_1=\frac{\partial}{\partial q_0},\qquad X^{\mathbb{H}}_2=q_0\frac{\partial}{\partial q_0}+q_1\frac{\partial}{\partial q_1}+q_2\frac{\partial}{\partial q_2}+q_3\frac{\partial}{\partial q_3},
\\
X^{\mathbb{H}}_3=2q_0\left(q_1\frac{\partial}{\partial q_1}+q_2\frac{\partial}{\partial q_2}+q_3\frac{\partial}{\partial q_3}\right)+(q_0^2-q_1^2-q_2^2-q_3^2)\frac{\partial}{\partial q_0},\\
\end{gathered}$$
which satisfy the same commutation relations as $X_1,X_2$ and $X_3$ in (\ref{ComSl2}). So, $X_1^{\mathbb{H}},X_2^{\mathbb{H}},X_3^{\mathbb{H}}$ span a real Lie algebra $V^{\mathbb{H}}$ isomorphic to $\mathfrak{sl}(2,\mathbb{R})$. Hence, the quaternionic Riccati equation \eqref{qric} with $t$-dependent real coefficients is an $\mathfrak{sl}(2,\mathbb{R})$-Lie system isomorphic to (\ref{sl2}) with respect to the Lie algebra isomorphism $\phi:V^{\mathbb{H}}\rightarrow V$ being given by
$$
\phi(X^{\mathbb{H}}_\alpha)=X_\alpha,\qquad \alpha=1,2,3.
$$
Moreover, we have that $\phi(X^\mathbb{H}_t)=X_t$ for every $t\in\mathbb{R}$. Hence, the symmetry system for (\ref{qric}) is (\ref{slsys}). If we assume for instance $f_0=k\in \mathbb{R}$, and $b_1(t)=\eta(t)$, $b_2(t)=0$ and $b_3(t)=1$, we obtain that (\ref{qric}) is isomorphic to the Lie system (\ref{ric2}). Hence, for certain values of $\eta(t)$ given in Table \ref{table1}, we can derive several Lie symmetries for quaternionic Riccati equations.

\subsubsection{The generalized Darboux--Brioschi--Halphen system}
We now show that our theory in particular, and the whole theory of Lie systems in general, can be used to study autonomous systems of first-order ODEs.
We consider the generalized Darboux--Brioschi--Halphen system (DBH system)\cite{Darboux}:
\begin{equation}\label{Partial}
\left\{
\begin{aligned}
\frac{{\rm d}w_1}{{\rm d}t}&=w_3w_2-w_1w_3-w_1w_2+\tau^2,\\
\frac{{\rm d}w_2}{{\rm d}t}&=w_1w_3-w_2w_1-w_2w_3+\tau^2,\\
\frac{{\rm d}w_3}{{\rm d}t}&=w_2w_1-w_3w_2-w_3w_1+\tau^2,\\
\end{aligned}\right.
\end{equation}
where
$$
\tau^2\equiv\alpha_1^2(\omega_1-\omega_2)(\omega_3-\omega_1)+\alpha_2^2(\omega_2-\omega_3)(\omega_1-\omega_2)+\alpha_3^2(\omega_3-\omega_1)(\omega_2-\omega_3)
$$
and $\alpha_1,\alpha_2,\alpha_3$ are real constants.

When $\tau=0$, system (\ref{Partial}) retrieves the classical DBH system solved by Halphen \cite{CH03,Darboux,Halphen} which appears in the study of triply orthogonal surfaces and the vacuum Einstein equations for hyper-K\"ahler Bianchi-IX metrics. For $\tau\neq 0$, the generalized DBH system can be considered as a
 reduction of the self-dual Yang--Mills equations corresponding to an infinite-dimensional gauge group of diffeomorphisms of a
three-dimensional sphere \cite{CH03}. 

It can be proved that (\ref{Partial}) is an $\mathfrak{sl}(2,\mathbb{R})$-Lie system. Indeed, it is the associated system to the $t$-dependent vector field
\begin{multline*}
    X_t^{\rm DBH}=(w_3w_2-w_1(w_3+w_2)+\tau^2)\frac{\partial}{\partial w_1} + (w_1w_3-w_2(w_1+w_3)+\tau^2)\frac{\partial}{\partial w_2}\\ + (w_2w_1-w_3(w_2+w_1)+\tau^2)\frac{\partial}{\partial w_3}=-X_3^{\rm DBH}\,.
\end{multline*}
This vector field span a Lie algebra $V^{\rm DBH}$ of vector fields along with
\[
    X_1^{\rm DBH}=\frac{\partial}{\partial w_1} + \frac{\partial}{\partial w_2} + \frac{\partial}{\partial w_3},\quad X_2^{\rm DBH}=w_1\frac{\partial}{\partial w_1} + w_2\frac{\partial}{\partial w_2} + w_3\frac{\partial}{\partial w_3}
\]
satisfying the commutation relations (\ref{ComSl2}). In consequence, $X_1^{\rm DBH},X_2^{\rm DBH}$ and $X_3^{\rm DBH}$ span a three-dimensional Lie algebra of vector fields $V^{\rm DBH}$ isomorphic to $\mathfrak{sl}(2,\mathbb{R})$ and then $X^{\rm DBH}$ is an $\mathfrak{sl}(2,\mathbb{R})$-Lie system. Since $X_1,X_2,X_3$ admit the same structure constants as (\ref{ComSl2}), the symmetry system for $X^{\rm DBH}$ becomes (\ref{slsys}) with $b_1(t)=b_2(t)=0$ and $b_3(t)=-1$, namely
\begin{equation*}
\frac{{\rm d}f_0}{{\rm d}t}=b_0(t),\qquad \frac{{\rm d}f_1}{{\rm d}t}=0,\qquad \frac{{\rm d}f_2}{{\rm d}t}=-2f_1,\qquad  \frac{{\rm d}f_3}{{\rm d}t}=-f_2-b_0(t).
\end{equation*}

 Hence, for $b_0(t)=0$, we obtain
 $$
 Y=t_0\frac{\partial}{\partial t}+\lambda_1 X_1-(2\lambda_1 t-\lambda_2)X_2+(\lambda_1t^2-\lambda_2t+\lambda_3)X_3,\qquad \lambda_1,\lambda_2,\lambda_3,t_0\in\mathbb{R}.
 $$
 Evidently, these vector fields span a Lie algebra of Lie symmetries isomorphic to $\mathfrak{sl}(2,\mathbb{R})$ for $t_0=0$. For $b_0(t)=c_0$, we obtain
 $$
 Y=(c_0t+t_0)\frac{\partial}{\partial t}+\lambda_1 X_1-(2\lambda_1 t-\lambda_2)X_2+(\lambda_1t^2-(\lambda_2+c_0)t+\lambda_3)X_3,\quad \lambda_1,\lambda_2,\lambda_3,t_0,c_0\in\mathbb{R}.
 $$
 Finally, for $b_0(t)=c_0t$, we get
 $$
 Y=\left(t_0+\frac{c_0t^2}2\right)\frac{\partial}{\partial t}+\lambda_1 X_1-(2\lambda_1 t-\lambda_2)X_2+\left[\lambda_1t^2-\left(\lambda_2+\frac{c_0t}2\right)t+\lambda_3\right]X_3,\,\, 
 $$
with $\lambda_1,\lambda_2,\lambda_3,t_0,c_0\in\mathbb{R}$. Since $\mathfrak{W}=\langle 1,t,t^2\rangle$ is a Lie algebra with respect to the Lie bracket, $\{\cdot,\cdot\}_\mathbb{R}$, we obtain in view of Proposition \ref{finite} that
 $$
 \mathcal{S}^V_{X,\mathfrak{W}}=\{Y\in \mathcal{S}^V_X|  f_0\in \langle 1,t,t^2\rangle\}
 $$
 is a Lie algebra of Lie symmetries. By choosing appropriately the constant coefficients of the above vector fields and setting $\tau=0$, we recover the Lie algebra of symmetries isomorphic to $\mathfrak{sl}(2,\mathbb{R})$ for classical DBH systems given by Nucci in \cite[p. 295]{Nu05}.
\subsection{Applications to systems of HODEs}

The study of systems of HODEs through Lie systems implies the addition of extra variables
in order to express a higher-order system as a first-order one. The introduction of these extra
variables frequently results in the obtainment of non-local Lie symmetries, as we shall exemplify in forthcoming examples.
Certain non-local Lie symmetries can be identified with the prolongations of certain vector fields.

\begin{lemma}
 Given a Lie algebra of Lie point symmetries $V$, its prolongations form a Lie algebra $\widehat{V}$ isomorphic to the former.
\end{lemma}
\begin{proof}
 Consider the map
\begin{equation*}
\begin{aligned}
\Phi:\Gamma({\rm T}(\mathbb{R}\times N))&\rightarrow \Gamma({\rm T}[\mathbb{R}\times {\rm T}^{p}N])\\
X&\mapsto \widehat{X},
\end{aligned}
\end{equation*}
mapping sections of ${\rm T}(\mathbb{R}\times N)$, i.e., vector fields on $\mathbb{R}\times N$, onto sections of ${\rm T}(\mathbb{R}\times {\rm T}^pN)$, i.e., vector fields on the manifold $\mathbb{R}\times{\rm T}^pN$, where ${\rm T}^pN$ is the so-called
{\it $p$-order tangent space} \cite{LM87}. Roughly speaking, ${\rm T}^pN$ is the space of equivalence classes of curves in $N$ with the same Taylor expansion up to order $p$. It admits a differentiable structure induced by the variables $x_1,\ldots,x_n$ of the coordinate systems of $N$ and the induced variables $x_i^{k)}$ describing the derivatives  in terms of $t$ up to order $p$ of the coordinates of a curve within $N$. Given the so-called {\it contact one-forms} $\theta^k_{i}={\rm d}x^{k)}_i-x_i^{k+1)}{\rm d}t$ with $i=1,\dots,n$ and $k=0,\ldots,p-1$ on $\mathbb{R}\times {\rm T}^pN$, we say that $\widehat{X}$ is the prolongation of $X\in \Gamma({\rm T}(\mathbb{R}\times N))$ to $\mathbb{R}\times {\rm T}^pN$ if
and only if every $\mathcal{L}_{\widehat X}\theta^k_i$ belongs to the contact ideal spanned by all the contact forms and $J^{p}_{\pi_{*}}\widehat{X}=X$, with $J^p_{\pi*}$ being the tangent map to the projection $J^p_\pi:(t,x_i,x_i^{k)})\in \mathbb{R}\times {\rm T}^pN\mapsto (t,x_i)\in\mathbb{R}\times N$. This implies that $\Phi$ is $\mathbb{R}$-linear and injective. Additionally,
\begin{equation*}
 \mathcal{L}_{[\widehat{X_1},\widehat{X_2}]}\theta^k_i=\left(\mathcal{L}_{\widehat {X_1}}\mathcal{L}_{\widehat {X_2}}-\mathcal{L}_{\widehat {X_2}}\mathcal{L}_{\widehat {X_1}}\right)\theta^k_i
\end{equation*}
belongs to the ideal of the contact forms because  $\mathcal{L}_{\widehat{X_1}}$ and  $\mathcal{L}_{\widehat{X_2}}$ do so.
Moreover,
$$
J^{p}_{\pi_{*}}[\widehat{X_1},\widehat{X_2}]=[J^{p}_{\pi_{*}}\widehat{X_1},J^{p}_{\pi_{*}}\widehat{X_2}]=[X_1,X_2].
$$
Hence, $[\widehat {X_1},\widehat {X_2}]$ must be the prolongation of $[X_1,X_2]$, i.e., $[\widehat {X_1},\widehat {X_2}]=\widehat{[X_1,X_2]}$. In this way, $\Phi$ is a Lie algebra morphism. Obviously, given a Lie algebra of Lie point symmetries $V$, its prolongations form a Lie algebra $\widehat{V}=\Phi(V)$ isomorphic to $V$.
\end{proof}
\begin{definition}
 We call $\widehat{V}$ the Lie algebra whose elements are the prolongations of the elements in ${V}.$
\end{definition}

\subsubsection{Second-order Kummer--Schwarz equation}

As an example of a second-order Lie system possessing a Vessiot--Gulberg Lie algebra isomorphic to $\mathfrak{sl}(2,\mathbb{R})$,
we present the second-order Kummer--Schwarz equation
\begin{equation}\label{KS1}
 \frac{{\rm d}^2x}{{\rm d}t^2}=\frac{3}{2x}\left(\frac{{\rm d}x}{{\rm d}t}\right)^2-2c_0x^3+2\eta(t)x,
\end{equation}
where $c_0$ is a real constant and $\eta(t)$ is an arbitrary $t$-dependent function. The interest of this differential equation is due to its appearance in the study of other remarkable (systems of) ODEs, e.g., Riccati and Ermakov systems, with physical and mathematical applications \cite{CGL11,Co94}.
It is well known that (\ref{KS1}) is a Lie system when it is rewritten as a first-order system by introducing an additional variable $v\equiv {\rm d}x/{\rm d}t$, i.e.,
\begin{equation}\label{KSsys}
\left\{\begin{aligned}
\frac{{\rm d}x}{{\rm d}t}&=v,\\
\frac{{\rm d}v}{{\rm d}t}&=\frac{3}{2}\frac{v^2}{x}-2c_0x^3+2\eta(t)x.
\end{aligned}\right.
\end{equation}
This system is associated with the $t$-dependent vector field
\begin{equation*}
 M^{\rm KS}=M_3^{\rm KS}+\eta(t)M_1^{\rm KS},
\end{equation*}
where the vector fields
\begin{equation*}
 M_1^{\rm KS}=2x\frac{\partial}{\partial v},\quad M_2^{\rm KS}=x\frac{\partial}{\partial x}+2v\frac{\partial}{\partial v},\quad M_3^{\rm KS}=v\frac{\partial}{\partial x}+\left(\frac{3}{2}\frac{v^2}{x}-2c_0x^3\right)\frac{\partial}{\partial v}
\end{equation*}
satisfy the same commutation relations as the vector fields $X_1,X_2,X_3$ in (\ref{ComSl2}) and they therefore span a Vessiot--Guldberg Lie algebra isomorphic to $\mathfrak{sl}(2,\mathbb{R})$.
Since the basis $M_1^{\rm KS},M_2^{\rm KS},M_3^{\rm KS}$ have the same structure constants as the $\mathfrak{sl}(2,\mathbb{R})$-Lie systems analyzed in Section \ref{LSSL2}, we can define, for instance, a Lie algebra morphism $\phi$ mapping $M_1^{\rm KS},M_2^{\rm KS},M_3^{\rm KS}$ to the basis $X^{\rm Ric}_1,X^{\rm Ric}_2,X^{\rm Ric}_3$ for Riccati equations (\ref{Ricc}). In such a case, $M^{\rm KS}$ maps to $X^{\rm Ric}$ and in view of Proposition \ref{SymTheo}, the symmetry system for $M^{\rm KS}$ is of the form of the symmetry system (\ref{ric2}). As a consequence, the particular solutions $f_0$, $f_1$, $f_2$ and $f_3$ for Riccati equations, detailed in Table \ref{table1}, are valid for (\ref{KSsys}) as well.

\subsection{Lie symmetries for Aff(R)-Lie systems}

In this section, we aim to obtain elements of $\mathcal{S}^{V}_X$ for Aff$(\mathbb{R})$-Lie systems. We choose a basis of vector fields $\{X_1,X_2\}$ of $V\simeq$ Aff$(\mathbb{R})$
with $[X_1,X_2]=X_1,$
and express the most general Lie system with Vessiot--Guldberg Lie algebra $V$ as $X=a(t)X_1+b(t)X_2.$
Let us now look for its Lie symmetries of the form
\begin{equation}\label{wzor}
Y=f_0(t)\frac{\partial}{\partial t}+f_1(t)X_1+f_2(t)X_2.
\end{equation}

The symmetry condition gives rise to a symmetry system
\begin{equation}\label{Asys}
\left\{\begin{aligned}
 \frac{{\rm d}f_0}{{\rm d}t}&=b_0(t),\\
\frac{{\rm d}f_1}{{\rm d}t}&=f_0\frac{{\rm d}a}{{\rm d}t}(t)+a(t)b_0(t)+b(t)f_1-a(t)f_2,\\
\frac{{\rm d}f_2}{{\rm d}t}&=f_0\frac{{\rm d}b}{{\rm d}t}(t)+b(t)b_0(t),
\end{aligned}\right.
\end{equation}
associated with the $t$-dependent vector field
$$\Gamma^{{\rm Aff}(\mathbb{R})}_X=b_0(t)Z_0+\frac{{\rm d} a}{{\rm d} t}W_1+\frac{{\rm d} b}{{\rm d} t}W_2+b_0(t)a(t)Z_1+b_0(t)b(t)Z_2+a(t)Y_1+b(t)Y_2$$
where 
\begin{equation*}\label{com4}
Y_1=-f_2\frac{\partial}{\partial f_1},\quad Y_2=f_1\frac{\partial}{\partial f_1},\quad W_1=f_0\frac{\partial}{\partial f_1},\quad W_2=f_0\frac{\partial}{\partial f_2},\quad Z_\alpha=\frac{\partial}{\partial f_\alpha},\quad \alpha=0,1,2.
\end{equation*}
Since, $[Y_1,Y_2]=Y_1$, then $V_L=\langle Y_1,Y_2\rangle$ gives rise to a Lie algebra isomorphic to Aff$(\mathbb{R}).$
Moreover, 
\begin{equation}\label{comm4}
\begin{gathered}
\left[Z_0,W_1\right]=Z_1,\quad [Z_0,W_2]=Z_2,\quad [Z_i,W_j]=0,\quad i,j=1,2,\quad [Z_\alpha,Z_\beta]=0,\quad \alpha,\beta=0,1,2,\\
[Y_1,W_1]=0,\quad [Y_1,W_2]=W_1,\quad [Y_1,Z_0]=0,\quad [Y_1,Z_1]=0,\quad [Y_1,Z_2]=Z_1,\\
[Y_2,W_1]=-W_1,\quad [Y_2,W_2]=0,\quad [Y_2,Z_0]=0,\quad [Y_2,Z_1]=-Z_1,\quad [Y_2,Z_2]=0.\\
\end{gathered}
\end{equation}
In this way, $A_1=\langle Z_0,Z_1,Z_2\rangle$ is an ideal of $A_1+A_2$, with $A_2=\langle W_1,W_2\rangle$  and $A_1\oplus_S A_2$ is an ideal of $A_1+A_2+V_L$. Hence, system (\ref{Asys}) possesses a Vessiot--Guldberg Lie algebra 
\begin{equation}\label{VG}
(A_1\oplus_S A_2)\oplus_S V_L\simeq (\mathbb{R}^3\oplus_S\mathbb{R}^2)\oplus_S{\rm Aff}(\mathbb{R}).
\end{equation}
We can solve (\ref{Asys}) when $b_0(t)=0$. In this case, $f_0=k\in \mathbb{R}$ and (\ref{Asys}) reduces to a  trivial equation for $f_2$ and a linear one for $f_1$. The general solution reads:

\begin{equation}\label{sol}
\begin{gathered}
f_1(t)=\left[\int_0^t{\left[k\frac{{\rm d}a}{{\rm d}t'}(t')-a(t')(kb(t')+c_1)\right]e^{-\int_0^{t'}{b(t''){\rm d}t''}}{\rm d}t'}+c_2\right]e^{\int_0^{t}{b(t'){\rm d}t'}},\\
f_2(t)=kb(t)+c_1,
\end{gathered}
\end{equation}
where $c_1$ and $c_2$ are constants of integration.

In order to illustrate Aff$(\mathbb{R})$-Lie systems through a physical example, we consider the Buchdahl equation \cite{Bu64,CSL05,CN10} 
$$
\frac{{\rm d}^2x}{{\rm d}t^2}=f(x)\left(\frac{{\rm d}x}{{\rm d}t}\right)^2+a_2(t)\frac{{\rm d}x}{{\rm d}t},
$$
where $f$ is any function depending on $x$ and $a_2$ is any $t$-dependent function. The interest of this system is due to its appearance in General Relativity \cite{CN10}. We add a new variable $v\equiv {\rm d}x/{\rm d}t$ to obtain
\begin{equation}\label{Buchdahl}
\left\{\begin{aligned}
\frac{{\rm d}x}{{\rm d}t}&=v,\\
\frac{{\rm d}v}{{\rm d}t}&=f(x)v^2+a_2(t)v,
\end{aligned}\right.
\end{equation}
the hereafter called {\it Buchdahl system}. It is known that (\ref{Buchdahl}) is a Lie system \cite{BBHLS13}. Indeed, (\ref{Buchdahl}) describes the integral curves of the $t$-dependent vector field
$$
X^{BD}=v\frac{\partial}{\partial x}+(f(x)v^2+a_2(t)v)\frac{\partial}{\partial
v}=X_1-a_2(t)X_2,
$$
where
$$
X_1^{BD}=v\frac{\partial}{\partial x}+ f(x)v^2\frac{\partial}{\partial v},\qquad
X_2^{BD}=-v\frac{\partial }{\partial v},
$$
satisfy $[X_1,X_2]=X_1$. That is, $X$ is an ${\rm Aff}(\mathbb{R})$-Lie system. By applying Theorem \ref{The:SLS}, we see that the Lie symmetries of (\ref{Buchdahl}) of the form
$$
Y=f_0(t)\frac{\partial}{\partial t}+f_1(t)X_1+f_2(t)X_2
$$
are determined by the first-order system (\ref{Asys}) with $a(t)=1$ and $b(t)=-a_2(t)$, i.e.,
\begin{equation}\label{Buchdahl2}
\left\{\begin{aligned}
\frac{{\rm d}f_0}{{\rm d}t}&=b_0(t),\\
\frac{{\rm d}f_1}{{\rm d}t}&=b_0(t)-a_2(t)f_1-f_2,\\
\frac{{\rm d}f_2}{{\rm d}t}&=-f_0\frac{{\rm d}a_2}{{\rm d}t}-b_0(t)a_2(t),
\end{aligned}\right.
\end{equation}
which is associated to the $t$-dependent vector field
\begin{equation}
\Gamma_X^{{\rm BD}}=b_0(t)Z_0-\frac{{\rm d}a_2}{{\rm d}t}W_2+b_0(t)Z_1-a_2(t)b_0(t)Z_2+Y_1-a_2(t)Y_2,
\end{equation}
where the vector fields $Z_0,Z_1, Z_2,Y_1,Y_2,W_2$ are those detailed in \eqref{com4} and have
the commutation relations \eqref{comm4}. Hence, these vector fields span a Lie algebra (\ref{VG}). Therefore, we can obtain the Lie symmetries for this system of the form (\ref{wzor}) by substituting $a(t)=1$ and $b(t)=-a_2(t)$ in (\ref{sol}).
\subsection{A Painlev\'e--Ince equation}
Let us study the Lie symmetries of the Painlev\'e--Ince equation \cite{AS94,Ince,KL09} 
\begin{equation}\label{Ricc3}
\frac{{\rm d}^2x}{{\rm d}t^2}=-3x\frac{{\rm d}x}{{\rm d}t}-x^3.
\end{equation}
This differential equation has been widely studied in many works (see \cite{CLSecOrd} and references therein) due to its interesting properties
and applications. From the point of view of the theory of Lie systems, it is a HODE Lie system, i.e., by adding
the new variable $v\equiv {\rm d}x/{\rm d}t$, we obtain a system
\begin{equation}\label{Ricc31}
\left\{\begin{aligned}
\frac{{\rm d}x}{{\rm d}t}&=v,\\
\frac{{\rm d}v}{{\rm d}t}&=-3xv-x^3,
\end{aligned}\right.
\end{equation}
which can be viewed as an $\mathfrak{sl}(3,\mathbb{R})$-Lie system with a Vessiot--Gulberg Lie algebra $V$ spanned by the vector fields (see \cite{CLSecOrd})
\begin{equation}\label{VF}
\begin{aligned}
X_1&=v\frac{\partial}{\partial x}-(3xv+x^3)\frac{\partial}{\partial v},\,\, &X_2&=\frac{\partial}{\partial  v},\\
X_3&=-\frac{\partial}{\partial x}+3x\frac{\partial}{\partial v},\,\, &X_4&=x\frac{\partial}{\partial x}-2x^2\frac{\partial}{\partial v},\\
X_5&=(v+2x^2)\frac{\partial}{\partial x}-x(v+3x^2)\frac{\partial}{\partial v},\,\, &X_6&=2x(v+x^2)\frac{\partial}{\partial  x}+2(v^2-x^4)\frac{\partial}{\partial v},\\
X_7&=\frac{\partial}{\partial x}-x\frac{\partial}{\partial v},\,\,
&X_8&=2x\frac{\partial}{\partial x}+4v\frac{\partial}{\partial v}.
\end{aligned}
\end{equation}
As a consequence, we can apply our theory to obtain Lie symmetries of this HODE Lie system. Moreover, as (\ref{Ricc31}) determines the integral curves of $X_1$, we can search for a vector field within $V$ commuting with $X_1$. Since $[X_1,X_6]=0$ (see \cite{CLSecOrd}), we obtain that $X_6$ is a Lie symmetry for (\ref{Ricc31}) and a non-local Lie symmetry for (\ref{Ricc3}).

\section{A class of Lie symmetries for PDE Lie systems}

Let us generalize in this section some of the previous results to the PDE Lie systems. We hereafter denote $t=(t_1,\dots,t_s)\in \mathbb{R}^s$ and $x=(x_1,\dots,x_n)\in\mathbb{R}^n$.

\begin{theorem}\label{the:MainPDE}Let $V$ be a finite-dimensional real Lie algebra of vector fields on $\mathbb{R}^n$ with a basis $X_{1},\ldots,X_r$.
Given a PDE Lie system
\begin{equation}\label{pdesys}
 \frac{\partial x_i}{\partial t_l}=\sum_{\alpha=1}^rb_{\alpha l}(t)X_{\alpha}(x),\qquad l=1,\ldots,s,\quad i=1,\dots,n,
\end{equation}
then
\begin{equation}\label{SymSpe}
Y(t,x)=\sum_{\beta=1}^rf_{\beta}(t)X_{\beta}(x)
\end{equation}
is a Lie symmetry of (\ref{pdesys}) if and only if $[\bar{X}^l_l,Y]=0$ for every $l$, where $\bar X^l_l$ is the $l$-autonomization given by (\ref{xx2}) of the $t$-dependent vector field
$$
X_l(t,x)\equiv \sum_{\alpha =1}^ rb_{\alpha l }(t)X_\alpha(x),\qquad l=1,\ldots,s.
$$
\end{theorem}
\begin{proof}
The coordinate systems $\{t_1,\ldots,t_s\}$ on $\mathbb{R}^s$ and $\{x_1,\ldots,x_n\}$ on $\mathbb{R}^n$ induce a coordinate system on the fiber bundle $J^1\pi\simeq \mathbb{R}^s\times \mathbb{R}^ n\times\mathbb{R}^{ns}$ with respect to the projection
$\pi:(t_l,x_j)\in\mathbb{R}^s\times\mathbb{R}^n\mapsto (t_l)\in \mathbb{R}^s$ of the form $t_l,x_j,x_{j,l}$, with $1\leq j\leq n$ and $1\leq l\leq s$.
A vector field $Y=\sum_{k=1}^n\eta_k(t,x)\partial/\partial x_k,$ $t\in \mathbb{R}^s$ and $x\in \mathbb{R}^n$, is a Lie symmetry of (\ref{pdesys}) if and only if
$$\widehat{Y}F^i_l=0,\qquad\qquad  i=1,\ldots, n,\qquad l=1,\ldots,s$$
on the submanifold $\mathcal{S}=\cap_{i=1}^n\cap_{l=1}^s (F^i_l)^{-1}(0)$, with $\widehat{Y}:J^1\pi\rightarrow {\rm T}(J^1\pi)$ being the prolongation of $Y$, namely
$$
\widehat{Y}=\sum_{k=1}^n \left[\eta_k\frac{\partial}{\partial x_k}+\sum_{q=1}^s\left(\frac{\partial \eta_k}{\partial t_q}+\sum_{j=1}^n\frac{\partial \eta_k}{\partial x_j}x_{j,q}\right)\frac{\partial}{\partial x_{k,q}}\right],
$$and
$$
\begin{gathered}
F_l^i=x_{i,l}-\sum_{\alpha=1}^rb_{\alpha l}(t)X_{\alpha}^i(x),\qquad i=1,\ldots,n,\qquad l=1,\ldots,s,\\
X_\alpha=\sum_{i=1}^nX_\alpha^i(x)\frac{\partial}{\partial x_i},\qquad\alpha=1,\ldots,r.
\end{gathered}
$$
By assumption, $\eta_k(t,x)=\sum_{\beta=1}^rf_{\beta}(t)X_{\beta}^k(x)$ for $k=1,\ldots,n$. So,
\begin{equation*}
 \widehat{Y}=\sum_{k=1}^n\sum_{\beta=1}^r\left[ f_{\beta}(t)X_{\beta}^k\frac{\partial}{\partial x_k}+\sum_{q=1}^s \left(\frac{\partial f_{\beta}(t)}{\partial t_q}X_{\beta}^k+\sum_{j=1}^nf_{\beta}(t)\frac{\partial X_{\beta}^k}{\partial x_j}x_{j,q}\right) \frac{\partial}{\partial x_{k,q}}\right].
\end{equation*}
Substituting this in $\widehat{Y}F^i_l=0$, we obtain
\begin{eqnarray*}
\sum_{k=1}^n \sum_{\beta=1}^r\left[f_{\beta}(t)X_{\beta}^k\left(-\sum_{\alpha=1}^rb_{\alpha l}(t)\frac{\partial X_{\alpha}^i}{\partial x_k}\right)+ \sum_{q=1}^s\left(\frac{\partial f_{\beta}(t)}{\partial t_q}X_{\beta}^k+\sum_{j=1}^n{f_{\beta}(t)\frac{\partial X_{\beta}^k}{\partial x_j}x_{j,q}}\right)\delta^{l}_q\delta^i_k\right]\\
=\sum_{\beta=1}^r\left[ -\sum_{k=1}^n \sum_{\alpha=1}^r f_{\beta}(t)b_{\alpha l}(t)X_{\beta}^k\frac{\partial X_{\alpha}^i}{\partial x_k}+\frac{\partial f_{\beta}(t)}{\partial t_l}X_{\beta}^i+\sum_{j=1}^n f_{\beta}(t)\frac{\partial X_{\beta}^i}{\partial x_j}x_{j,l}\right]=0.
\end{eqnarray*}
Restricting the above expression to the submanifold $\mathcal{S}=\cap_{i=1}^n\cap_{l=1}^s (F^i_l)^{-1}(0)$ and renaming indexes appropriately, we obtain
\begin{equation*}
 \sum_{\beta=1}^r \left[\frac{\partial f_{\beta}(t)}{\partial t_l}X_{\beta}^i+ \sum_{k=1}^n \sum_{\alpha=1}^r\left(f_{\beta}(t)\frac{\partial X_{\beta}^i}{\partial x_k}b_{\alpha l}(t)X_{\alpha}^k-f_{\beta}(t)b_{\alpha l}(t)X_{\beta}^k\frac{\partial X_{\alpha}^i}{\partial x_k}\right)\right]=0.
\end{equation*}
Hence,
\begin{equation*}
\sum_{\beta=1}^r \left[ \frac{\partial f_{\beta}(t)}{\partial t_l}X_{\beta}^i+\sum_{\gamma=1}^r\sum_{k=1}^n f_{\beta}(t)b_{\gamma l}(t)\left(\frac{\partial X_{\beta}^i}{\partial x_k}X_{\gamma}^k-X_{\beta}^k\frac{\partial X_{\gamma}^i}{\partial x_k}\right)\right]=0,
\end{equation*}
whose right-hand side becomes, for each fixed $l$, the coefficients in the basis $\partial/\partial x^i$, with $i=1,\ldots,n$, of $[\bar X^l_l,Y]$. So, the above amounts to
\begin{multline*}
 \left[\bar{X}^l_l,Y\right]=\left[\frac{\partial}{\partial t_l}+\sum_{k=1}^n \sum_{\gamma=1}^r b_{\gamma l}(t)X_{\gamma}^k\frac{\partial}{\partial x_k},\sum_{i=1}^n \sum_{\beta=1}^r f_{\beta}(t)X_{\beta}^i \frac{\partial}{\partial x_i}\right]\\=\left[\frac{\partial}{\partial t_l}+\sum_{\gamma=1}^r b_{\gamma l}(t)X_\gamma,\sum_{\beta=1}^r f_{\beta}(t)X_\beta\right]=0.
\end{multline*}
Then, $Y$ is a Lie symmetry of (\ref{pdesys}) if and only if the condition
\begin{equation}\label{symcond2}
 \left[\bar{X}^l_l,Y\right]=0
\end{equation}
is satisfied for $l=1,\ldots,s.$
\end{proof}

\begin{theorem}\label{MTPDELieSystems}
Given a Lie symmetry of the form (\ref{SymSpe}) for the system (\ref{pdesys}), the
coefficients $f_1(t),\dots,f_r(t)$ satisfy a PDE Lie system admitting a Vessiot--Guldberg Lie algebra $V^{S}\simeq V/Z(V)$, where we recall that $V$ is a Vessiot--Guldberg Lie algebra for (\ref{pdesys}).
\end{theorem}
\begin{proof}
Let $X_1,\ldots,X_r$ be a basis for $V$ with structure constants $c_{\alpha\beta\gamma}$. From Theorem \ref{the:MainPDE}, we have
\begin{equation*}
\begin{aligned}
\left[\bar{X}^l_l,Y\right]&=\left[\frac{\partial}{\partial t_l}+\sum_{\alpha=1}^{r}{b_{\alpha l}(t)X_{\alpha}},\sum_{\delta=1}^r f_{\delta}(t)X_{\delta}\right]=\sum_{\delta=1}^r{\left({\frac{\partial f_{\delta}}{\partial t_l}X_{\delta}}+\sum_{\alpha=1}^r{b_{\alpha l}(t)f_{\delta}[X_{\alpha},X_{\delta}]}\right)}\\
&=\sum_{\pi=1}^r{\left(\frac{\partial f_\pi}{\partial t_l}+\sum_{\alpha=1}^r \sum_{\delta=1}^r b_{\alpha l}(t)f_{\delta}c_{\alpha \delta \pi}\right)X_\pi}=0.
\end{aligned}
\end{equation*}
Since $X_1,\dots,X_r$ are linearly independent over $\mathbb{R}$ and the coefficients of the above expression are only $t$-dependent, we get that the above amounts to
\begin{equation}\label{PDESymSys}
\frac{\partial f_\pi}{\partial t_l}=\sum_{\alpha,\delta=1}^r b_{\alpha l}(t)f_{\delta}c_{\delta\alpha  \pi},\qquad \pi=1,\ldots,r,\quad l=1,\ldots,s.
\end{equation}

To prove that this is a PDE Lie system, we define the vector fields
$$
Y_\alpha=\sum_{\delta,\pi=1}^rc_{\delta \alpha \pi}f_\delta\frac{\partial}{\partial f_\pi},\qquad \alpha=1,\ldots,r.
$$
We have already proved that $[Y_\alpha,Y_\beta]=\sum_{\delta=1}^rc_{\alpha\beta\delta}Y_\delta$ in Theorem \ref{The:SLS}. So, these vector fields span a Lie algebra isomorphic to $V/Z(V)$ (for a proof of this fact follow the same line of reasoning as in Theorem \ref{The:SLS}).
In terms of these vector fields, we see that (\ref{PDESymSys}) is related to the $t$-dependent vector fields $X_l(t,x)=\sum_{\alpha=1}^rb_{\alpha l}(t)X_\alpha(x)$, with $l=1,\ldots,s$. Additionally, to be a PDE Lie system, the above system (\ref{PDESymSys}) must satisfy
the condition
\begin{equation}\label{conint}
\sum_{\alpha,\mu=1}^r\left(\frac{\partial b_{\alpha \pi}}{\partial \kappa}-\frac{\partial b_{\alpha \kappa}}{\partial \pi}+\sum_{\delta,\epsilon=1}^rb_{\delta \kappa}b_{\epsilon \pi}c_{\delta\epsilon\alpha}\right)f_\mu c_{\mu\alpha\sigma}=0,\qquad \kappa\neq \pi=1,\ldots,r,\qquad \sigma=1,\ldots,r.
\end{equation}
The expression in brackets vanishes due to the integrability condition for system (\ref{pdesys}). Hence, (\ref{PDESymSys}) is a PDE Lie system. We call it the {\it symmetry system} for (\ref{pdesys}) relative to $V$.
\end{proof}

\begin{definition}
 Given a PDE Lie system $X$ with a Vessiot--Guldberg Lie algebra $V$, we call $\mathcal{S}_X^{V}$
the space of Lie symmetries of $X$ that are also $t$-dependent vector fields taking values in $V$.
\end{definition}

We can straightforwardly prove that the space $\mathcal{S}_X^{V}$ is a Lie algebra.

\section{Lie symmetries for $\mathfrak{sl}(2,\mathbb{R})$-PDE Lie systems}
An $\mathfrak{sl}(2,\mathbb{R})$-PDE Lie system is a PDE Lie system admitting a Vessiot--Guldberg Lie algebra isomorphic to $\mathfrak{sl}(2,\mathbb{R}).$
Let us obtain the elements of $\mathcal{S}_X^V$ for this case. 

Let us choose a basis of vector fields $\{X_1,X_2,X_3\}$ for $V$ satisfying the same commutation relations as in \eqref{ComSl2}.
Let us write a general PDE Lie system whose autonomization for a fixed value $l$ is
\begin{equation*}
 \bar{X}^l_l=\frac{\partial}{\partial t_l}+b_{1l}(t)X_1+b_{2l}(t)X_2+b_{3l}(t)X_3,\quad t=(t_1,\ldots,t_s) \in \mathbb{R}^s,\quad 1\leq l\leq s
\end{equation*}
and a certain type of possible Lie symmetry $Y=f_1(t)X_1+f_2(t)X_2+f_3(t)X_3$, where $f_1(t),f_2(t),f_3(t)$ are $t$-dependent functions to be determined by the symmetry condition (\ref{symcond2}).
This leads us
to the system of $s$ first-order PDEs
\begin{equation}
\left\{\begin{aligned}
&\frac{\partial f_1}{\partial t_l}=b_{2l}(t)f_1-b_{1l}(t)f_2,\\
 &\frac{\partial f_2}{\partial t_l}=2(b_{3l}(t)f_1-b_{1l}(t)f_3),\\
 & \frac{\partial f_3}{\partial t_l}=b_{3l}(t)f_2-b_{2l}(t)f_3,
\end{aligned}\right.
\end{equation}
with $l=1,\ldots,s$. Expressed in terms of $t$-dependent vector fields, for a fixed value of $l$
\begin{equation*}
\Gamma^{\mathfrak{sl}(2,\mathbb{R})}_{l}=b_{1l}(t)Y_{1}+b_{2l}(t)Y_{2}+b_{3l}(t)Y_{3}, \quad l=1,\ldots,s,
\end{equation*}
where
\begin{equation}\label{eq2}
 Y_{1}=-f_2\frac{\partial}{\partial f_1}-2f_3\frac{\partial}{\partial f_2},\quad Y_{2}=f_1\frac{\partial}{\partial f_1}-f_3\frac{\partial}{\partial f_3},\quad Y_{3}=2f_1\frac{\partial}{\partial f_2}+f_2\frac{\partial}{\partial f_3}
\end{equation}
close the commutation relations in (\ref{ComSl2}). So, they span a Lie algebra isomorphic to $\mathfrak{sl}(2,\mathbb{R})$.

\subsection{The partial Riccati equation}
Let us consider the partial Riccati equation, i.e., the PDE system
\begin{equation}\label{PDERicc}
 \frac{\partial x}{\partial t_1}=b_{11}(t)+b_{21}(t)x+b_{31}(t)x^2,\qquad
\frac{\partial x}{\partial t_2}=b_{12}(t)+b_{22}(t)x+b_{32}(t)x^2,
\end{equation}
with the $t$-dependent coefficients satisfying the appropriate integrability condition (\ref{conint}).
Such systems appear in the study of WZNW equations and multidimensional Toda systems \cite{FGRSZ99}. Observe that the partials $\partial x/\partial t_1$ and $\partial x/\partial t_2$ are related to the $t$-dependent vector fields
${X}^{\rm {pRic}}_{t_1}=b_{11}(t)X_{1}^{\rm {pRic}}+b_{21}(t)X_{2}^{\rm {pRic}}+b_{31}(t)X_{3}^{\rm {pRic}}$ and
${X}^{\rm {pRic}}_{t_2}=b_{12}(t)X_{1}^{\rm {pRic}}+b_{22}(t)X_{2}^{\rm {pRic}}+b_{32}(t)X_{3}^{\rm {pRic}}$,
with $$X_{1}^{\rm {pRic}}=\frac{\partial}{\partial x},\quad X_{2}^{\rm {pRic}}=x\frac{\partial}{\partial x},\quad X_{3}^{\rm {pRic}}=x^2\frac{\partial}{\partial x}$$
satisfying the commutation relations (\ref{ComSl2}). That is, the vector fields $\langle X_{1}^{\rm {pRic}},X_{2}^{\rm {pRic}},X_{3}^{\rm {pRic}}\rangle$ span a Vessiot--Guldberg Lie algebra $V^{\rm {pRic}}\simeq \mathfrak{sl}(2,\mathbb{R}).$ Since we assume that the functions $b_{ij}(t)$ with $i=1,2,3$ and $j=1,2$ satisfy (\ref{conint}), we get that (\ref{PDERicc}) is a PDE Lie system.

Let us look for Lie symmetries of the form $Y=f_1(t)X^{\rm pRic}_{1}+f_2(t)X^{\rm pRic}_{2}+f_3(t)X^{\rm pRic}_{3}$ for (\ref{PDERicc}). In view of Theorem \ref{MTPDELieSystems}, such Lie symmetries are solutions of the system of PDEs
\begin{equation}
\left\{\begin{aligned}
&\frac{\partial f_1}{\partial t_j}={b_{2j}}(t)f_1-{b_{1j}}(t)f_2,\\
&\frac{\partial f_2}{\partial t_j}=2({b_{3j}}(t)f_1-{b_{1j}}(t)f_3),\\
&\frac{\partial f_3}{\partial t_j}={b_{3j}}(t)f_2-{b_{2j}}(t)f_3,
\end{aligned}\right.
\end{equation}
 with $j=1,2.$
This resulting system can be interpreted in terms of the $t$-dependent vector fields $\Gamma^{\rm{pRic}}_j=b_{1j}(t)Y_{1}+b_{2j}(t)Y_{2}+b_{3j}(t)Y_{3}$, with $j=1,2$ and (\ref{eq2}).
These vector fields have the same structure constants as the $X^{\rm pRic}_1,X^{\rm pRic}_2,X^{\rm pRic}_3$. Therefore, \eqref{PDERicc} is a PDE Lie system with a Vessiot--Guldberg Lie algebra isomorphic to $\mathfrak{sl}(2,\mathbb{R}).$
\subsection{Flat $\mathfrak{g}$-valued forms and their generators}
Let $G$ be a Lie group with Lie algebra $\mathfrak{g}$  and let $L_g:G\rightarrow G$ be the left-action of $G$ on $G$. We write $\omega$ for a one-form on an $n$-dimensonal manifold $N$ taking values in  $\mathfrak{g}$ and $\varphi:N\rightarrow G$ for a function on $N$ taking values in $G$. Consider the system of PDEs 
$$
\omega=L_{\varphi^{-1}*}{\rm d}\varphi
$$
determining the hereafter called generator $\varphi$ of $\omega$, where
\begin{equation}\label{conomega}
{\rm d}\omega+\omega\wedge\omega=0,
\end{equation}
and $\omega\wedge \omega$ stands for the wedge product of the Lie algebra-valued $\omega$ with itself. Equations of this type appear for instance in the study multidimensional Riccati type equations and Lie algebra valued flat connection forms.

Taking a local basis of coordinates $\{x_1,\ldots,x_n\}$ on $N$ and writing $\omega=\sum_{\mu=1}^n\omega_\mu(x) {\rm d}x_\mu$, where $\omega_\mu:N\rightarrow \mathfrak{g}$, we get
$$
\frac{\partial \varphi}{\partial x_\mu}=L_{\varphi*} \omega_\mu,\qquad \mu=1,\ldots,n.
$$
Choosing a basis $v_1,\ldots,v_r$ for $\mathfrak{g}$, we can write $\omega_\mu=\sum_{\alpha=1}^rf_{\alpha\mu}(x)v_\alpha$ for $\mu=1,\ldots,n$ and certain functions $f_{\mu\alpha}(x)$ with $\alpha=1,\ldots,r$ on $N$. Hence,
$$
\frac{\partial \varphi}{\partial x_\mu}=L_{\varphi*}\left( \sum_{\alpha=1}^rf_{\alpha\mu}(x)v_\alpha\right)=\sum_{\alpha=1}^rf_{\alpha\mu}(x)L_{\varphi*}v_\alpha, \qquad \mu=1,\ldots,n.
$$ 
Note that $L_{\varphi*} v_\alpha$ can be considered as the value at $\varphi\in G$ of the left-invariant vector field on $G$ given by $X^L_\alpha(\varphi)=L_{\varphi*} v_\alpha$. In consequence, we have
\begin{equation}\label{PDELieGroup}
\frac{\partial \varphi}{\partial x_\mu}=\sum_{\alpha=1}^rf_{\alpha\mu}(x)X^L_\alpha(\varphi), \qquad \mu=1,\ldots,n.
\end{equation}
Additionally, in view of (\ref{conomega}), we have that
$$
\frac{\partial f_{\alpha\mu}}{\partial x_\beta}-\frac{\partial f_{\alpha\beta}}{\partial x_\mu}+\sum_{\pi,\delta=1}^rf_{\pi\beta}f_{\delta\mu}c_{\pi\delta\alpha}=0, \quad \mu\neq \beta=1,\ldots,n,\quad \alpha=1,\ldots,n,
$$
and hence, (\ref{PDELieGroup}) becomes a PDE Lie system.

In terms of the basis $X^L_1,\ldots, X^L_r$ for $\mathfrak{g}$, which have structure constants $c_{\alpha\beta\gamma}$, we see that the symmetry system for (\ref{PDELieGroup}) reads as (\ref{PDESymSys}).
\section{Conclusions and Outlook}

We have constructed a natural family of Lie symmetries for Lie systems that can be described, studied and derived through the theory of Lie systems. In particular, we have shown that
these Lie symmetries can be obtained by solving another Lie system possessing a Vessiot--Guldberg Lie algebra related to the Vessiot--Guldberg Lie algebra of the Lie system under study. Our results have been generalized to the realm of HODEs and PDE Lie systems. Our procedures enable us to find Lie symmetries of isomorphic (PDE) Lie systems simultaneously.

To illustrate our theory, we have studied Lie symmetries for some systems of differential equations of physical and mathematical relevancy. As a result, we have retrieved and generalized several results concerning Lie symmetries of several types of known and new Riccati equations from a more geometrical point of view.
Moreover, we have shown that autonomous systems can also be studied through Lie systems.

In the future, we aim to introduce a more general perspective on the quaternionic Riccati equation
by imposing that the coefficients take values in the quaternionic field. Additionally, we aim to 
study the generalization of the quaternionic Riccati equation to other types of algebras and to study their geometric properties. 
Moreover, we also plan to perform a deeper study on PDE Lie systems.
All these things shall be included in forthcoming papers.

\section*{Acknowledgements}
J. de Lucas acknowledges partial financial support by research project MTM2010-12116-E. C. Sard\'on acknowledges a fellowship from the
University of Salamanca. P.G. Est\'evez and C. Sard\'on acknowledge partial financial support from research project SA226U13 by the JCyL. J. de Lucas and C. Sard\'on acknowledge partial financial support from the program HARMONIA of Polish National Science Centre
under the contract number DEC-2012/04/M/ST1/00523. F.J.~Herranz acknowledges partial financial support from
  the Spanish MINECO  under grant    MTM2010-18556 (with EU-FEDER support).

\end{document}